\documentclass[sigplan,screen]{acmart}
\copyrightyear{2027}
\acmYear{2027}
\setcopyright{cc}
\setcctype{by}
\acmConference[ASPLOS '27]{Proceedings of the 32nd ACM International Conference on Architectural Support for Programming Languages and Operating Systems, Volume 1}{April 11--15, 2027}{Heraklion, Greece}
\acmBooktitle{Proceedings of the 32nd ACM International Conference on Architectural Support for Programming Languages and Operating Systems, Volume 1 (ASPLOS '27), April 11--15, 2027, Heraklion, Greece}
\acmDOI{10.1145/3845814.3850059}
\acmISBN{979-8-4007-3012-2/2027/04}

\usepackage{array}

\DeclareRobustCommand{\code}[1]{\texttt{\detokenize{#1}}}
\newlength{\tabletextwidth}
\newcommand{\sysname}{TIRx}

\begin{document}

\title[\sysname{}: A Unified Tile Primitive ML Compiler]{\sysname{}: A Unified Tile Primitive ML Compiler Across Heterogeneous Backends}

\author{Bohan Hou}
\affiliation{%
  \institution{Carnegie Mellon University}
  \city{Pittsburgh}
  \state{Pennsylvania}
  \country{USA}
}
\email{bohanhou@andrew.cmu.edu}

\author{Hongyi Jin}
\affiliation{%
  \institution{Carnegie Mellon University}
  \city{Pittsburgh}
  \state{Pennsylvania}
  \country{USA}
}
\email{hongyij@andrew.cmu.edu}

\author{Guanjie Wang}
\affiliation{%
  \institution{Carnegie Mellon University}
  \city{Pittsburgh}
  \state{Pennsylvania}
  \country{USA}
}
\email{guanjiew@andrew.cmu.edu}

\author{Shushi Hong}
\affiliation{%
  \institution{Carnegie Mellon University}
  \city{Pittsburgh}
  \state{Pennsylvania}
  \country{USA}
}
\email{tlopexh@andrew.cmu.edu}

\author{Jinqi Chen}
\affiliation{%
  \institution{NVIDIA}
  \city{Santa Clara}
  \state{California}
  \country{USA}
}
\email{kathrync@nvidia.com}

\author{Yaxing Cai}
\affiliation{%
  \institution{NVIDIA}
  \city{Santa Clara}
  \state{California}
  \country{USA}
}
\email{yaxingc@nvidia.com}

\author{Lijie Yang}
\affiliation{%
  \institution{Princeton University}
  \city{Princeton}
  \state{New Jersey}
  \country{USA}
}
\email{ly3223@princeton.edu}

\author{Zihao Ye}
\affiliation{%
  \institution{NVIDIA}
  \city{Santa Clara}
  \state{California}
  \country{USA}
}
\email{zihaoy@nvidia.com}

\author{Yaoyao Ding}
\affiliation{%
  \institution{NVIDIA}
  \city{Toronto}
  \state{Ontario}
  \country{Canada}
}
\email{yaoyaod@nvidia.com}

\author{Ruihang Lai}
\affiliation{%
  \institution{Carnegie Mellon University}
  \city{Pittsburgh}
  \state{Pennsylvania}
  \country{USA}
}
\email{ruihangl@cs.cmu.edu}

\author{Todd C. Mowry}
\affiliation{%
  \institution{Carnegie Mellon University}
  \city{Pittsburgh}
  \state{Pennsylvania}
  \country{USA}
}
\email{tcm@cs.cmu.edu}

\author{Tianqi Chen}
\affiliation{%
  \institution{Carnegie Mellon University}
  \city{Pittsburgh}
  \state{Pennsylvania}
  \country{USA}
}
\affiliation{%
  \institution{NVIDIA}
  \city{Santa Clara}
  \state{California}
  \country{USA}
}
\email{tqchen@cmu.edu}
\renewcommand{\shortauthors}{Bohan Hou et al.}

\begin{abstract}
Modern ML kernels evolve rapidly with new optimizations and hardware features. ML compiler abstractions must accommodate richer memory hierarchies, more forms of thread cooperation, and increasingly heterogeneous backends while retaining the control required for peak performance. We present \sysname{}, an ML compiler for GPUs and AI accelerators that exposes these ingredients through explicit interfaces. \sysname{} introduces an explicit execution-scope abstraction that captures participating hardware teams at multiple granularities. It also uses tensor layouts to map logical tensors to heterogeneous hardware spaces through named axes and shard/replica/offset semantics. Finally, it introduces tile-primitive dispatch, which lowers high-level operations to efficient implementations according to scope, layout, and backend. This design preserves native-level control while automating recurring address calculation, operand preparation, and instruction-loop generation. Our evaluation shows vendor-level performance across the latest GPUs and AI accelerators, while native-code case studies show substantial reductions in low-level implementation code.
\end{abstract}


\maketitle

\section{Introduction}

\begin{figure*}[!t]
  \centering
  \includegraphics[width=\textwidth]{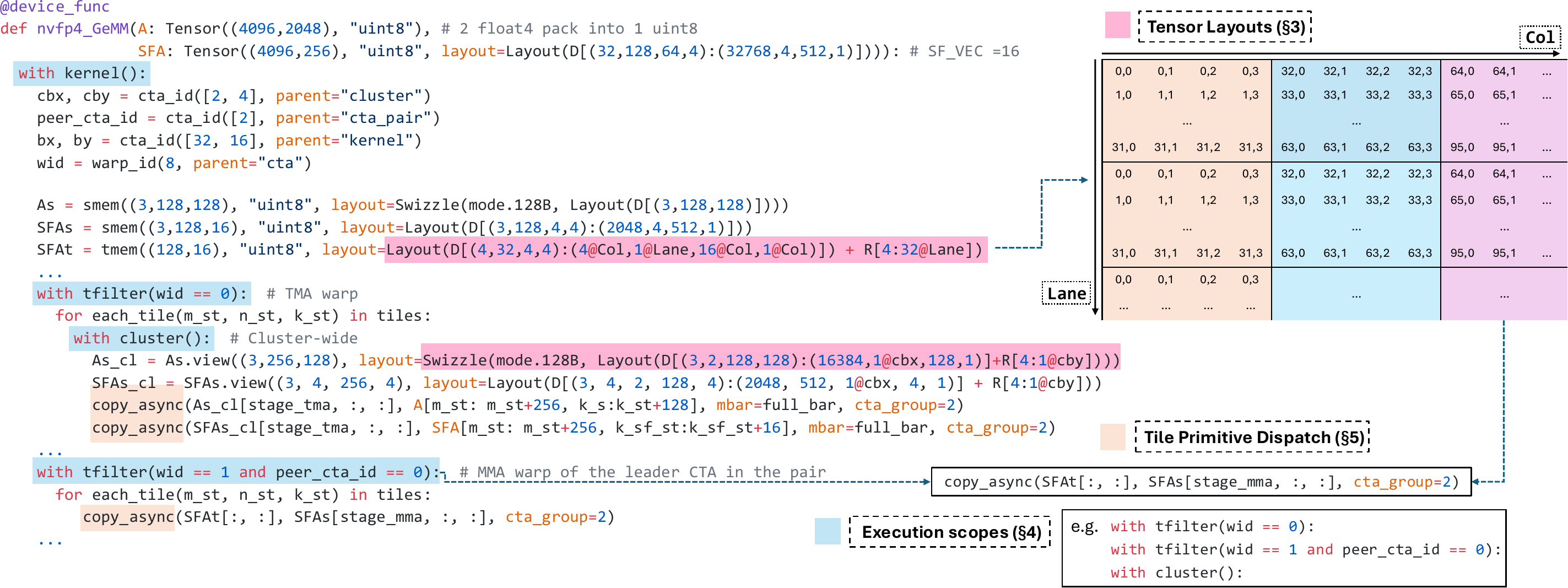}
  \Description{Overview of a \sysname{} NVFP4 GEMM kernel showing native-level kernel code together with tensor-layout, execution-scope, and tile-primitive interfaces.}
  \caption{
    \textbf{\sysname{} programming model on an NVFP4 GEMM kernel.}
    The kernel retains native-level control flow and scheduling.
    \textbf{Tensor layouts} describe logical-to-physical placement,
    \textbf{execution scopes} identify participating hardware teams, and
    \textbf{tile primitives} invoke reusable operations under these annotations.
    \S\ref{sec:axe}, \S\ref{sec:execution-scope}, and
    \S\ref{sec:tile-primitive-dispatch} describe the three interfaces.
    }
  \label{fig:overview}
\end{figure*}

Machine learning (ML) kernels have become a critical systems bottleneck for modern AI workloads. Attention and large language model (LLM) prefill and decode kernels, low-precision block-scaled GEMMs, fused MoE layers, and larger fused kernels~\cite{shah2024flashattention,zadouri2026flashattention,ye2025flashinfer,guo2025sonicmoe,cheng2025mirage,jin2026eventtensorunifiedabstraction} all directly determine end-to-end performance in training and serving.
The community has developed various compiler systems and kernel DSLs~\cite{tillet2019triton,CuTeDSL,ding2025tilus,wang2025tilelang} to help programmers express and compile ML kernels through reusable higher-level abstractions, and their abstractions must simultaneously address three challenges.

\textbf{Memory hierarchy.}
Modern kernels no longer operate only over conventional global/shared/register memory patterns.
The latest GPU~\cite{nvidia2025blackwell} features introduce richer memory spaces and access patterns, such as distributed shared memory (DSMEM) and multi-dimensional tensor memory (TMEM).
We see similar patterns of multi-dimensional memories in specialized accelerators~\cite{aws-trainium,google-tpu}.
These features often require more structured storage layouts; for example, block-scaled GEMMs may need scale factors to be stored in tensor memory (2D scratchpad with \code{Lane}$\times$\code{Col} dimensions) with replication patterns along its \code{Lane}.

\textbf{Thread hierarchy.}
The participating hardware team of a primitive is no longer limited to a full warp or a thread block.
Modern hardware introduces warpgroups, CTA pairs, and thread-block clusters, and many instructions have logical scope wider than the participant thread set itself.
For example, cluster-wide TMA multicasting involves one participating thread from leader CTAs in the cluster, and Blackwell \code{tcgen05} instructions may be issued by a thread from only one CTA in a CTA pair.

\textbf{Hardware heterogeneity.}
AI accelerators often expose execution and memory models that differ fundamentally from GPUs.
Accelerator kernels expose a single-thread programming model together with
on-chip buffers that require explicit partition and layout management.
A compiler must therefore handle not only richer GPU scopes and memory spaces, but also backend-specific execution and storage models.

Although these challenges appear in different forms, we observe that they repeatedly revolve around a shared set of ingredients.
First, a system needs an \textbf{execution scope} that specifies the hardware team that owns or participates in an operation.
Second, it needs a \textbf{tensor layout} that describes how logical tensors are mapped into increasingly rich and heterogeneous physical storage spaces.
Third, it needs \textbf{tile primitives together with an accompanying dispatch mechanism} that determine how an operation such as copy, matrix multiply, reduction, or arithmetic should be lowered under a given scope, layout, and backend.
Together, these ingredients define how compilers express and lower evolving kernel patterns across hardware.

Existing compiler systems each focus on a subset of these aspects.
Triton/Gluon~\cite{tillet2019triton,gluon} centers its programming model around CTA-level operators, and then extends it with warp- and warpgroup-level mechanisms.
Its layout system remains constrained by power-of-two shape assumptions~\cite{linear}. For example, DeepGEMM~\cite{deepgemm}'s $128 \times 224$ TMEM scale-factor tile already falls outside Triton's linear layout space.
TileLang~\cite{wang2025tilelang} and Tilus~\cite{ding2025tilus} exhibit a similar programming model, leaving cluster-wide TMA multicast support to roll out.
Their layouts do not account for memory replication semantics, while the aforementioned scale-factor tile also requires replication along the \texttt{Lane} axis.
CuTeDSL~\cite{CuTeDSL} provides \code{copy} and \code{gemm} primitives, but primarily at the per-thread level: users still need to specify thread-level work assignment through explicit thread/data partitioning.
ThunderKittens~\cite{stanfordThunderKittens20}, in contrast, standardizes a curated primitive family around \(16\times16\)-scale warp tiles, a simple layout space, and warp- and warpgroup-level primitives.
Pallas~\cite{pallas} supports warpgroup tiles and related primitives, but its layout design remains separated across TPU and GPU. As a result, programmers often face an unfavorable choice when a tile primitive does not fit a compiler's thread hierarchy, memory hierarchy, or hardware-target prerequisites: either wait for compiler support or fall back to native-level programming.

In this paper, we propose \textbf{\sysname{}}, a simple yet effective machine learning compiler
for ML kernels across GPUs and AI accelerators.
\sysname{} introduces explicit interfaces for the following key ingredients.
\textbf{Execution scope} is generalized from fixed warp- or thread-block-centric teams to richer team kinds and participation structures, allowing emerging cooperation patterns such as cluster cooperation to be represented directly.
\textbf{Tensor layout} is generalized from shape-constrained or backend-specific layout families to a uniform mapping from logical tensors with general shapes to heterogeneous hardware spaces with explicit named axes and shard/replica structure.
\textbf{Tile primitives and dispatch} are generalized through an explicit scope- and layout-aware mechanism that dispatches each primitive to an efficient implementation on the target backend, optionally with user guidance.
Together, these interfaces enable \sysname{} to expose efficiently lowerable tile primitives in cases where existing compilers cannot support them without forcing users back to native-level programming.

Importantly, the compiler infrastructure of \sysname{} is deliberately simple.
It keeps native-level code as the basis for kernel programming and adds only a thin, library-like compiler layer above it: programmers write native kernels and invoke tile primitives, while the compiler dispatches each primitive according to its scope, layout, and backend.
This preserves exact control and avoids compiler-pass modifications when backend features evolve. Reusable tile primitive interfaces automate recurring address calculation, operand preparation, and instruction-loop generation for common collective operations. Our contributions are as follows:
\begin{itemize}
\item We present \sysname{} with explicit execution scope for richer team kinds and participation structures.
\item We introduce tensor layout for mapping logical tensors to heterogeneous hardware spaces with shard/replica semantics, and scope- and layout-aware dispatch for lowering primitives to efficient implementations on the target backend.
\item We build an end-to-end compiler supporting a broad range of kernels across GPUs and AI accelerators.
\end{itemize}
In evaluation, \sysname{} reaches near-vendor performance and outperforms existing compiler solutions for FP16 GEMM, low-precision block-scaled GEMMs, and FlashAttention on NVIDIA B200.
For multi-GPU kernels, \sysname{} achieves up to \(1.40\times\) speedup on GEMM + Reduce-Scatter.
On AWS Trainium~1, \sysname{} matches or exceeds vendor-provided kernels. Our native-code case studies also show substantial reductions in low-level
implementation code.

\section{System Overview}

\sysname{} introduces a domain-specific frontend language (DSL) in Python together with a compiler that translates the DSL into target backend code.
\autoref{fig:overview} illustrates the overall design on an NVFP4 GEMM kernel.

For context, NVIDIA GPUs organize threads into warps and warpgroups and
support cooperation across CTAs and clusters. TMA provides asynchronous
global/shared-memory tile transfers, while Blackwell TMEM provides on-chip
storage for \code{tcgen05} instructions. Their participating hardware teams
and operand placements are captured by \sysname{}'s execution scopes and
tensor layouts.

As hardware evolves rapidly, it is important for programmers to retain native-level access.
\sysname{} builds on a native-level intermediate representation (IR)---expressions, control flow, and intrinsic calls---with a lowering infrastructure that maps these constructs directly to the target backend language. Above this native layer, \sysname{} introduces only a thin set of reusable tile primitive interfaces.
\textbf{Execution scope} specifies which hardware team owns or participates in an operation. \textbf{Tensor layout} explicitly attaches logical-to-physical placement to each tensor through declarations and library constructors.
\textbf{Tile primitive dispatch} consumes these layouts together with the primitive configuration and participating scope to select a registered backend implementation.
Together, these interfaces factor recurring collective operations into reusable compiler-managed pieces while preserving native-level kernel structure.

In the following sections, we describe these three interfaces in detail. Since all three are grounded in \sysname{}'s layout system, we begin with tensor layout.

\section{\sysname{} Layout}
\label{sec:axe}

\subsection{Overview}

\label{sec:axe-overview}

\begin{figure}[t]
  \centering
  \includegraphics[width=\columnwidth]{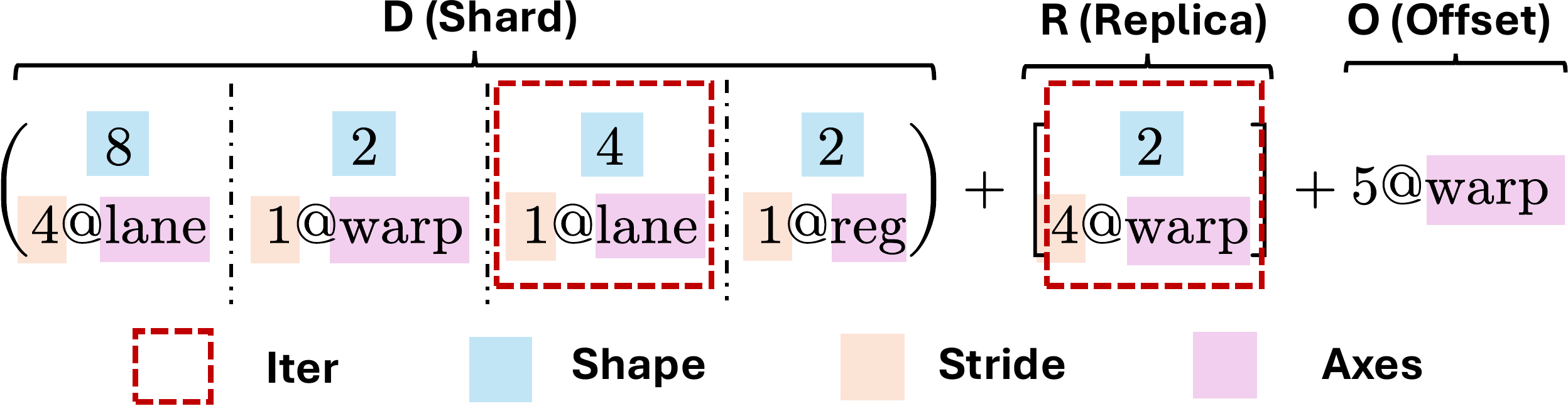}
  \Description{Diagram showing the three components of a \sysname{} layout: shard D as a list of iters, replica R as a set of iters, and offset O as a fixed coordinate.}
  \caption{Elements of \sysname{} Layout. An \emph{iter} specifies a triple \((\text{extent}, \text{stride}, \text{axis})\) and defines a linear, strided access on that axis. A list of iters forms the shard part \(\mathbf{D}\), a set of iters forms the replica part \(\mathbf{R}\), and \(\mathbf{O}\) is a fixed offset.}
  \label{fig:axe-elements}
\end{figure}

Tensor layout is used to represent how a logical tensor is stored in physical resources. \sysname{} extends the classical \emph{shape--\allowbreak stride} model of tensor layout.
In NumPy~\cite{numpy} or PyTorch~\cite{paszke2019pytorch}, a dense layout is described by a shape and strides (each stride is the memory step when the corresponding index increases).
\sysname{} generalizes it by allowing strides to be semantically \emph{named} and bound to different \emph{axes} that represent hardware resources, including memory, threads, and devices.
Since these axes are used for sharding the logical shape, we name this \textbf{D} (shard).
Additionally, to support replication and constant offset along hardware axes,
\sysname{} further introduces \textbf{R} (replica) and \textbf{O} (offset), respectively.
In formal terms, a \sysname{} layout maps a logical index to a \emph{set} of coordinates on named axes. We decompose this mapping into three components~(\autoref{fig:axe-elements}):

\textbf{D (Shard).} A list of one or more \emph{iters}, each with an extent and a stride on some axis. \(\mathbf{D}\) partitions the logical index across these iters and produces a base coordinate. This generalizes shape–stride to multiple axes. We write the \(\mathbf{D}\) iter list in parentheses.

\textbf{R (Replica).} A set of replication iters that enumerate offsets in hardware space, independent of the logical index. Adding each element of this set to the \(\mathbf{D}\) result yields replication or broadcasting. We write the \(\mathbf{R}\) iter set in square brackets.

\textbf{O (Offset).} A fixed coordinate offset (one integer per axis) added to every result. This places data at a specific base position or reserves exclusive resources; unused coordinates arise naturally because the map is set-valued.

Formally, let $L$ denote a \sysname{} layout mapping. For a given logical index $x$, the layout produces:
\[
L(x) = \bigl\{ D(x) + r + O \mid r \in R \bigr\},
\]
where $D(x)$ is the base coordinate tuple obtained from the sharded iters, $r$ ranges over all combinations of the replication iters (if $R$ is empty, we interpret this as a single zero offset), and $O$ is the constant offset vector. By construction, $L(x)$ can be a singleton or contain multiple coordinates.

\textbf{Familiar strided layouts.}
The same representation covers conventional layouts:
$(64,128):(136@\texttt{m},1@\texttt{m})$ represents a
$64\times128$ tensor with row pitch 136, while
$(2,64):(1@\texttt{m},2@\texttt{m})$ interleaves two
length-64 streams element-wise.

\begin{figure*}[!t]
    \centering
    \includegraphics[width=1\textwidth]{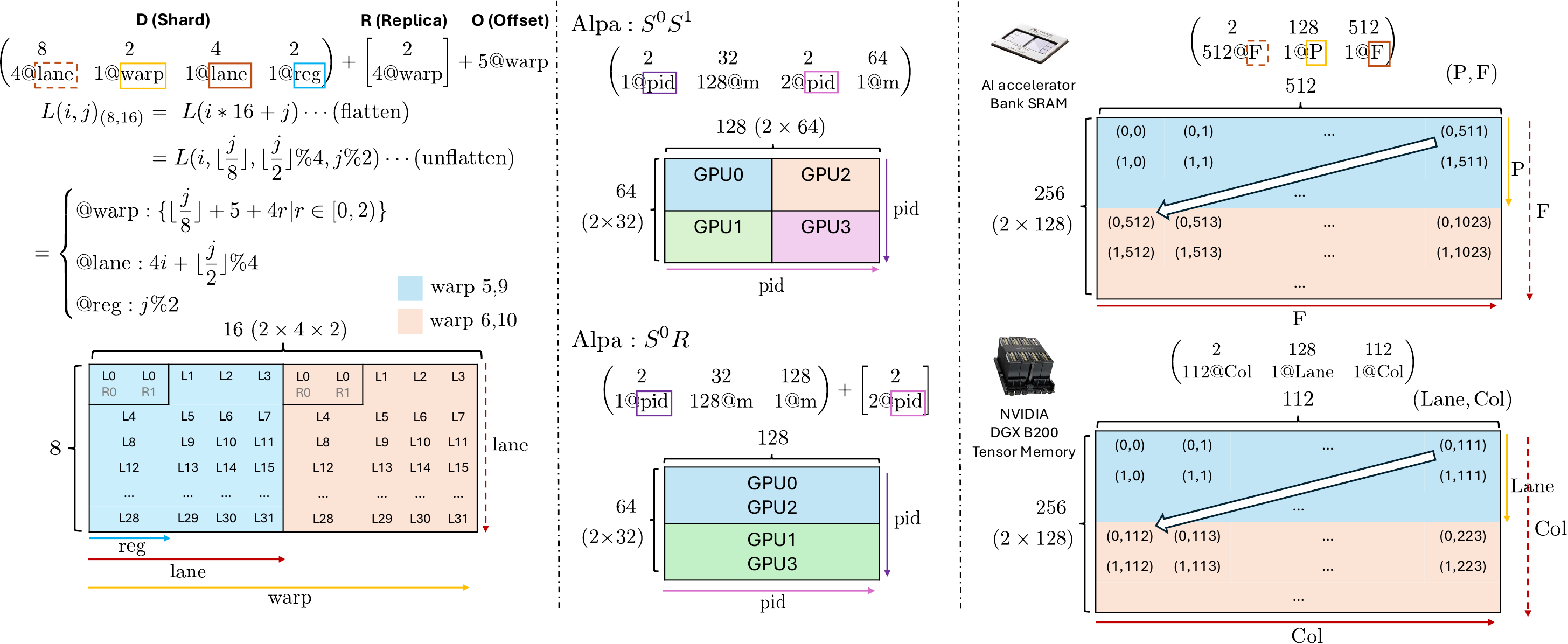}
    \Description{Three layout examples: a GPU tensor core tile mapped to warps, lanes, and registers; distributed sharding across a 4-GPU mesh; and native hardware memory layouts for an AI accelerator scratchpad and Blackwell tensor memory.}
    \caption{
\textbf{\sysname{} layout case studies.} The figures in each column show the logical shape of the tensor and each cell shows the physical location(s) mapped from the corresponding logical element. Axes are color coded.
\textbf{Left.} Mapping a logical $8\times 16$ tile to 4 GPU warps with 32 lanes and 2 registers each; 2 warps sharded and 2 warps replicated.
\textbf{Middle.} Distributed sharding of a $64\times 128$ matrix across 4 GPUs; the top uses full sharding across 4 GPUs, the bottom uses a $2\times 2$ mesh with shards and replicas.
\textbf{Right.} Native hardware memories; the top depicts an AI accelerator 2D-partitioned scratchpad SRAM, the bottom shows NVIDIA Blackwell 2D tensor memory.
}
\label{fig:layout-overview}
\end{figure*}

\subsection{\sysname{} Layout Case Studies}

\label{sec:axe-examples}

\begin{figure}[t]
  \centering
  \includegraphics[width=\columnwidth]{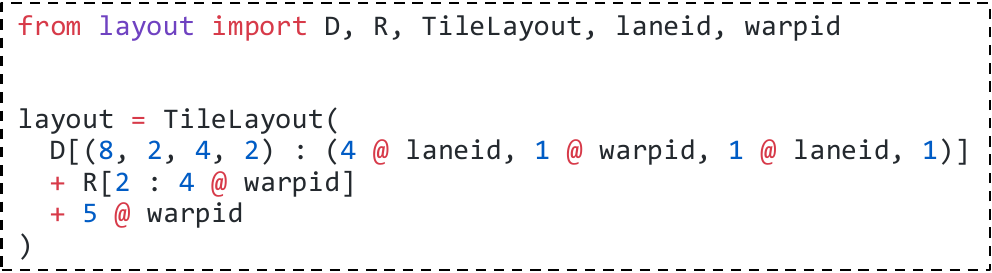}
  \Description{Code snippet showing the \sysname{} Layout Python API for defining layouts with shard, replica, and offset components.}
  \caption{\sysname{} Layout Python API. If some stride is not paired with an axis, the axis \texttt{m} is used by default.}
  \label{fig:layout-python-api}
\end{figure}

\autoref{fig:layout-overview} gives examples of \sysname{} layouts in practice, illustrating mappings in three scenarios: (A) within a GPU warp’s registers, (B) across multiple GPUs in a distributed mesh, and (C) into specialized memory structures. \autoref{fig:layout-python-api} shows how layouts are defined in the Python DSL. We next walk through these examples.

\textbf{NVIDIA Tensor Core tile.}
Consider a logical tile $L(i,j)$ of shape $(8,16)$ that we want to map into a GPU kernel’s thread and register space. Specifically, we distribute this tile across 2 warps, with 32 lanes each, such that each lane holds a portion of the tile in its registers, matching the specification of NVIDIA’s tensor-core instructions. We represent this in \sysname{} layout as:

\emph{Shard (D)}: $\begin{pmatrix}
8 & 2 & 4 & 2 \\
4@\texttt{lane} & 1@\texttt{warp}  & 1@\texttt{lane}  & 1@\texttt{reg}    \\
\end{pmatrix}$.
The shard part denotes that we factor the logical indices into iters of extents $8$, $2$, $4$, and $2$, distributed over axes \texttt{lane}, \texttt{warp}, \texttt{lane}, and \texttt{reg}, respectively.
The original tile has two dimensions: row and column, with extents 8 and 16, respectively.
The row dimension is assigned to the iter with stride 4 and axis \texttt{lane}. The column dimension is split into 3 sub-dimensions ($2 \times 4 \times 2$) and assigned to \texttt{warp}, \texttt{lane}, and \texttt{reg} axes, respectively, with stride 1 for all axes.

\emph{Replica (R)}: $\begin{bmatrix}
  2 \\
  4@\texttt{warp}
\end{bmatrix}$.
This indicates that the entire tile is replicated twice across the \texttt{warp} axis, with a stride of 4 warps between the two replicas. In effect, we now have warps $\{0,1\}$ and $\{4,5\}$ each holding an identical copy of the $8\times 16$ tile.

\emph{Offset (O)}: $5@\texttt{warp}$. This adds an offset of 5 to the \texttt{warp} axis. Thus, we will have warps $\{5,6\}$ and $\{9,10\}$ hold the $8\times 16$ tile.

\textbf{Distributed sharding on a 2$\times$2 GPU mesh.}
Now suppose we have a $64\times 128$ tensor that we want to distribute across 4 GPUs arranged in a $2\times2$ mesh (\autoref{fig:layout-overview}). Label the device IDs as:
\(\begin{bmatrix}\text{GPU0} & \text{GPU2}\\ \text{GPU1} & \text{GPU3}\end{bmatrix}\). We can express different layouts for distributing the tensor:

\emph{Fully sharded}: Split the 64 rows across mesh rows and split the 128 columns across mesh columns. The \sysname{} layout:
\[
\begin{pmatrix}
  2 & 32 & 2 & 64 \\
  1@\texttt{pid} & 128@\texttt{m} & 2@\texttt{pid} & 1@\texttt{m}
\end{pmatrix},
\]
where the first and third factors (2) are on the device axis \texttt{pid} and the remaining factors (32 and 64) are on the memory axis \texttt{m}.

\emph{Shard with replication}: Now split the rows across the two mesh row groups, and replicate each row shard to both GPUs in that group. In \sysname{}, we represent it as:
\begin{align*}
  \begin{pmatrix}
  2 & 32 & 128 \\
  1@\texttt{pid} & 128@\texttt{m} & 1@\texttt{m}
\end{pmatrix}
  +
  \begin{bmatrix}
  2 \\
  2@\texttt{pid}
\end{bmatrix}.
\end{align*}

\begin{figure}[!t]
  \centering
  \includegraphics[width=\columnwidth]{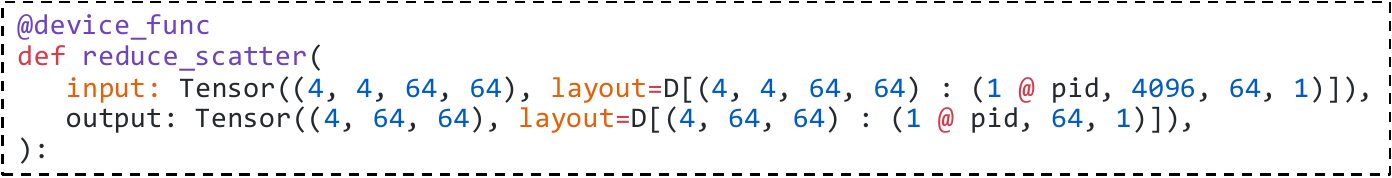}
  \Description{Diagram showing a distributed tensor signature for a reduce-scatter operation with input and output sharding specifications.}
  \caption{Example distributed tensor signature for reduce scatter.}
  \label{fig:layout-dtensor}
\end{figure}

These layouts encode common parallelism strategies. For example, in Alpa~\cite{alpa}'s notation, they represent $S^0S^1$ and $S^0R$ sharding specs, respectively. Because \sysname{} layout naturally supports distributed execution, we can use it to represent the
distributed sharding constraints of a distributed tensor~\cite{xu2021gspmd, ptdtensor}. \autoref{fig:layout-dtensor} shows an example of how
we can use \sysname{} to represent a reduce-scatter kernel that accepts a DTensor with shape $(4, 4, 64, 64)$ that shards over the first dimension,
and reduces over dimension $1$, generating an output DTensor with shape $(4, 64,64)$ that shards over the first dimension.
The compiler will generate runtime checks to ensure the consistency of the input DTensor and the declared \sysname{} layout.

\textbf{Native multi-dimensional memory in accelerators.} AI accelerators’ on-chip scratchpad buffers are organized using a multi-dimensional addressing scheme (dimensions denoted as P for memory bank partitions and F for free dimensions).
Suppose we have a logical tensor that spans 128 partitions and uses a 2D tiling of dimensions $256\times512$. A \sysname{} layout might be:
\[
\begin{pmatrix}
  2 & 128 & 512 \\
  512@\texttt{F} & 1@\texttt{P} & 1@\texttt{F}
\end{pmatrix}.
\]

NVIDIA’s Blackwell GPUs introduce a dedicated tensor memory with native 2D addressing (think of it as a matrix of size \texttt{Lane} $\times$ \texttt{Col}). A tensor placed in this memory might have a layout like:
\[
\begin{pmatrix}
  2 & 128 & 112 \\
  112@\texttt{Col} & 1@\texttt{Lane} & 1@\texttt{Col}
\end{pmatrix}.
\]
This layout tiles the tensor across the 2D plane so that it spans 128 lanes and 224 columns. DeepGEMM~\cite{deepgemm}'s FP8 block-scaled GEMM uses such a non-power-of-two tile size because tensor memory cannot hold two accumulator stages together with the scale factors if the tile size is 256.

\subsection{\sysname{} Layout Design Rationale}

Tensor layout has recently received significant attention in the ML compiler field~\cite{linear,cecka2026cute,ding2025tilus}. We discuss the rationale of our layout design, which specifically serves the needs of the \sysname{} compiler.

\textbf{General shape support.} Non-power-of-two shapes are universal in global memory tensors and common in intra-kernel tensors (e.g., multi-stage SMEM buffers). On Blackwell GPUs and AI accelerators, on-chip scratchpad memories are multi-dimensional, and non-power-of-two tile sizes arise naturally when scratchpad capacity is limited.

\textbf{Logical to physical formulation.} One alternative way to formalize layouts is to map physical locations to logical coordinates, such that replication (one logical element stored in multiple physical
locations) can still be defined as point-valued functions. However, for tensors that span physical locations in a strided pattern, some physical locations might not have a well-defined mapping.

\textbf{Explicit hardware resource labels.} Another alternative is to map logical coordinates to a generic \((\texttt{t}, \texttt{m})\) pair, leaving the precise meanings of \texttt{t} and \texttt{m} to be recovered from context. The difficulty is that the same layout expression can then denote different physical placements depending on the surrounding execution context. For example, the layout
\[
\begin{pmatrix}
  128 & 4 \\
  1@\texttt{t} & 1@\texttt{m}
\end{pmatrix}
\]
could refer to a tensor distributed over the thread IDs of an entire thread block, or to a tensor distributed over thread IDs local to a subteam such as a warpgroup. Disambiguating such cases would require the compiler to consult additional contextual information carried by the tensor, or to rely on extra conventions in the programming model---for example, that the meaning of \texttt{t} is inherited from the execution scope at the tensor's definition site.

\sysname{} avoids baking such restrictions into the programming model. In the two cases above, the layouts become
\[
\begin{pmatrix}
  128 & 4 \\
  1@\texttt{tid} & 1@\texttt{m}
\end{pmatrix},
\quad
\begin{pmatrix}
  128 & 4 \\
  1@\texttt{tid\_in\_wg} & 1@\texttt{m}
\end{pmatrix},
\]
This design allows users to construct tensors with the layouts they want, while leaving legality and feasibility checks to tile primitive dispatch.

These labels are backend-defined typed physical coordinate spaces. Their
extents come from the launch configuration, the hardware resource, or the
tensor allocation, and backend rules determine where each label can appear.
\autoref{tab:hardware-axes} summarizes representative backend-defined labels used by \sysname{}.

\begin{table}[t]
\centering
\small
\caption{Representative hardware axes used by \sysname{} layouts.}
\label{tab:hardware-axes}
\setlength{\tabcolsep}{3pt}
\settowidth{\tabletextwidth}{%
  \begin{tabular}{@{}lll@{}}
    \code{cbx/cby/cbz} & distributed & launch-defined
  \end{tabular}%
}
\setlength{\tabletextwidth}{\dimexpr\columnwidth-\tabletextwidth-2\tabcolsep\relax}
\begin{tabular}{@{}l l l >{\raggedright\arraybackslash}p{\tabletextwidth}@{}}
\toprule
Axis & Backend & Extent & Use \\
\midrule
\code{m}
  & all & buffer-defined & linear memory coordinate \\
\code{laneid}
  & NVIDIA & 32 & lane within a warp \\
\code{warpid}
  & NVIDIA & launch-defined & warp within a CTA \\
\code{wid_in_wg}
  & NVIDIA & 4 & warp within a warpgroup \\
\code{tid_in_wg}
  & NVIDIA & 128 & thread within a warpgroup \\
\code{cbx/cby/cbz}
  & NVIDIA & launch-defined & CTA coordinate in a cluster \\
\code{Lane}
  & Blackwell & 128 & TMEM Lane coordinate \\
\code{Col}
  & Blackwell & 512 & TMEM Column coordinate \\
\code{P}
  & Trainium & 128 & partition dimension \\
\code{F}
  & Trainium & tile-defined & free dimension \\
\code{pid}
  & distributed & mesh-defined & device coordinate \\
\bottomrule
\end{tabular}
\end{table}

\section{Execution Scope}
\label{sec:execution-scope}

\subsection{Active Thread Set and Team Kind}

To precisely define the thread team that cooperatively executes a tile primitive, the compiler must reason about two orthogonal properties at every program point: which threads are executing, and at what granularity an
operation logically occurs. We call these the \emph{active thread set} and the \emph{team kind}, respectively.

Consider three operations from a warp-specialized GEMM kernel~\cite{deepgemm} running on a $2\times4$ cluster with 2-CTA cooperative tensor cores (\autoref{fig:scope-example}):

\begin{figure}[t]
  \centering
  \includegraphics[width=\columnwidth]{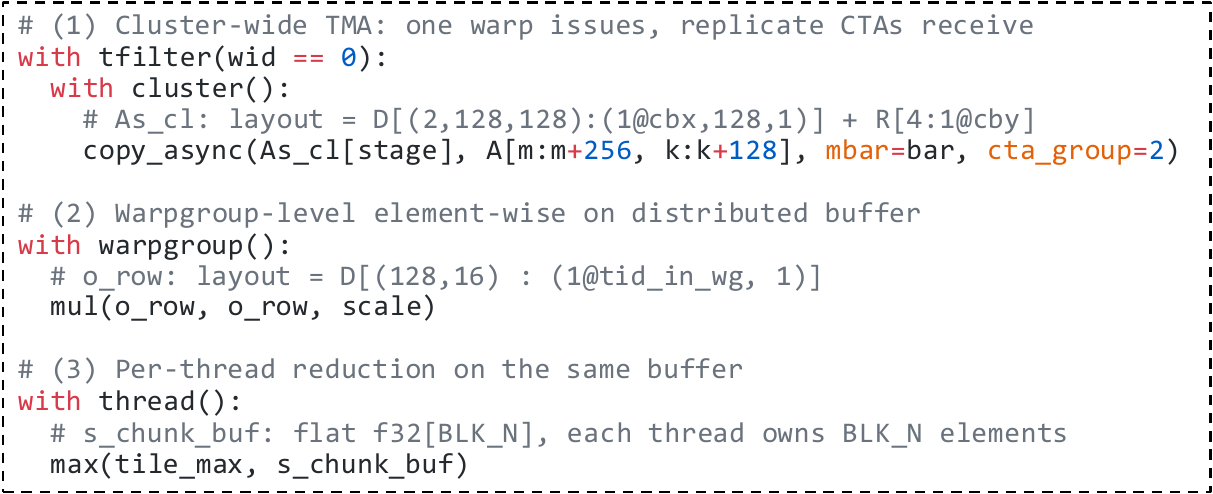}
  \Description{Diagram showing the \sysname{} DSL execution scope slice API with thread filtering and team switching constructs.}
  \caption{Each operation has a distinct team kind (cluster, warpgroup, thread) and involved tensor layout, while the active thread set is independently controlled by filtering.}
  \label{fig:scope-example}
\end{figure}

In~(1), the team kind is cluster: 8~CTAs cooperate on a partitioned load, with the buffer's layout
declaring that rows are sharded across \texttt{cbx} and replicated across \texttt{cby}. The active thread set is filtered to warp~0 via
\texttt{tfilter}, but the cluster-level team kind is unchanged.
In~(2), the team kind is warpgroup: the buffer's layout distributes 128 rows across 128~threads, and all participate in the element-wise multiply.
In~(3), the team kind switches to thread via \texttt{with thread()}; each thread independently reduces its own flat buffer of \texttt{BLK\_N} elements.

We define the \textbf{execution scope} at each program point as consisting of two components: (1)~the \textbf{active thread set}~$A$, which is the set of threads \textit{physically} executing the current region,
and (2)~the \textbf{team kind}~$\kappa$,  which determines the granularity at which the current operation is \textit{logically} cooperatively executed. \sysname{} exposes two source constructs that \textit{independently} modify these:

\textbf{Filtering} (\texttt{tfilter}) narrows the active thread set. For example, \texttt{tfilter(wid == 0)} restricts execution to warp~0 in each CTA, but the operation
remains at whatever team kind was in effect.

\textbf{Team switch} (e.g., \texttt{with warpgroup()}) changes the team kind. A \texttt{with thread():} block tells the compiler
that subsequent operations should be lowered at per-thread granularity, while the same set of active threads continues executing. Team kinds follow each backend's execution model. NVIDIA provides warp, warpgroup, CTA, and cluster scopes, while Trainium tile primitives use \code{thread()} scope under its single-thread programming model.

In general, handling arbitrary active thread set patterns at every program point would require the compiler to reason about complex thread predicates. \sysname{} imposes three restrictions to keep the representation tractable:

\begin{itemize}
    \item
Filtering can only shrink the active thread set.
    \item
Filter conditions must be a contiguous range on a single axis (e.g., \texttt{tfilter(warp[2:4])}) or a single-point predicate (e.g., \texttt{tfilter(wid == 0)}).
    \item
A team switch requires that every non-empty team instance has the same intra-team thread pattern.
\end{itemize}

For example, an active warp range of $[0, 16)$ in the thread block decomposes cleanly into 4~warpgroups of 4~warps each. A range $[3, 7)$ straddles two warpgroup boundaries---warpgroup~0 gets 1~warp, warpgroup~1 gets 3---producing non-uniform teams, so the switch is rejected at compile time.

These restrictions apply only to representing hardware-team participation for collective tile primitives; they do not restrict general tensor indices, masks, predicates, or native-level control flow in a \sysname{} kernel. These restrictions mirror how hardware exposes cooperation: teams (warp, warpgroup, CTA, cluster) are themselves uniform and contiguous on a single axis, and instructions under certain team-kind require every instance to have the same intra-team shape. Patterns that violate these restrictions do not map naturally to existing hardware cooperation mechanisms.

\subsection{Layout Represents Execution Scope}

The active thread set must support two operations efficiently: (a)~narrowing via filters, and (b)~decomposition into inter-team and intra-team components when a team switch occurs.
Under these conditions, our layout system can represent the active thread set. For example,
\[
A = \begin{pmatrix}
  l & w & c \\
  1\text{@lane} & 1\text{@warp} & 1\text{@cta}
\end{pmatrix}
+ (l_0, w_0, c_0)
\]

\noindent where $(l, w, c)$ are the extents and $(l_0, w_0, c_0)$ are the offsets on each axis. The total active thread count is $l \times w \times c$. At kernel entry with 512~threads per CTA in a 4-CTA cluster:
\[
A_{\text{entry}} = \begin{pmatrix}
  32 & 16 & 4 \\
  1\text{@lane} & 1\text{@warp} & 1\text{@cta}
\end{pmatrix}
\]

\paragraph{Filtering.} A filter narrows the range on one axis by intersecting with the specified interval.

For example, \texttt{tfilter(warp[0:1])} produces:
\[
A' = \begin{pmatrix}
  32 & 1 & 4 \\
  1\text{@lane} & 1\text{@warp} & 1\text{@cta}
\end{pmatrix}
\]

\noindent Only warp~0 in each CTA remains active. The team kind is unchanged.

\paragraph{Team switch and decomposition.} A team switch decomposes $A$ into two sub-layouts by partitioning the axes according to the target team kind. The
\emph{inter-team} layout enumerates the active team instances; the \emph{intra-team} layout describes the thread pattern within each team.

For example, switching to warp scope from $A_{\text{entry}}$ yields:
\[
A_{\text{intra}} = \begin{pmatrix}
  32 \\
  1\text{@lane}
\end{pmatrix},
\quad
A_{\text{inter}} = \begin{pmatrix}
  16 & 4 \\
  1\text{@warp} & 1\text{@cta}
\end{pmatrix}
\]

\noindent producing 64~independent warp teams of 32~lanes each.

Switching to warpgroup scope requires \emph{factoring} the warp axis into warp-within-warpgroup and warpgroup-ID. Given $A$ with
warp extent~$w$ and offset~$w_0$, factoring succeeds when $w_0 \bmod 4 = 0$ and $w \bmod 4 = 0$ (aligned case), or when $w \leq 4 - (w_0 \bmod 4)$ (single-warpgroup case). For the aligned case with $w=16$:
\begin{align*}
A_{\text{intra}} &= \begin{pmatrix}
  32 & 4 \\
  1\text{@lane} & 1\text{@warp\_in\_wg}
\end{pmatrix},
\\
A_{\text{inter}} &= \begin{pmatrix}
  4 & 4 \\
  1\text{@wg\_id} & 1\text{@cta}
\end{pmatrix}
\end{align*}

\noindent producing 16~warpgroup teams (4~per CTA), each containing 128~threads (32~lanes $\times$ 4~warps). A misaligned warp range that straddles warpgroup boundaries unevenly is rejected at compile time.

Execution scope participates directly in the dispatch key and legality
contract. Each tile primitive receives the current team kind together with
the decomposed $(A_{\mathrm{inter}}, A_{\mathrm{intra}})$. For example, a
\code{cast} under warpgroup scope resolves to the registered
warpgroup-vectorized rule, whereas thread scope resolves to the scalar
per-element rule. The active-thread-set layout supplies the participation
pattern checked by the selected rule.

\subsection{Expressiveness Boundary}
\label{sec:expressiveness-boundary}

Irregular indices and unequal amounts of useful work remain compatible with
tile-primitive dispatch when the participating hardware team is statically
representable and each operand tile admits a \sysname{} layout. This covers
masks and padding, indptr-based tiling, and causal masking. Backend-supported
indexed accesses are also expressible; for example,
\code{copy_async(As[:, :], Ag[indexer[:], :])} can lower to
\code{tma.gather4} when legal. These cases motivate two abstraction
boundaries: runtime-dependent team membership or ownership, and
position-dependent strides.

\textbf{Runtime-dependent team membership or ownership.} Collective tile-primitive dispatch requires a statically representable participating team and ownership mapping. Dynamic spatial partitioning, work stealing, and related schemes may determine team membership or tile ownership at runtime, placing that scheduling decision outside this dispatch contract.

\textbf{Position-dependent strides.} A \sysname{} layout supports dynamic strides whose value is uniform across the corresponding logical dimension. A layout whose stride varies with the logical position---for example, a different row stride for each row---falls outside this representation.

\sysname{}'s native-level programming layer remains available around these abstractions. A kernel may implement dynamic scheduling or irregular address computation directly and continue to use tile primitives for compatible operations in the same kernel. This composition keeps the scope, layout, and tile-primitive interfaces applicable to regular subcomputations within otherwise irregular kernels.

\section{Tile Primitive Dispatch}
\label{sec:tile-primitive-dispatch}

Tile primitive dispatch is local and registry-based. For each primitive call, the primitive kind, backend, operand storage roles, execution scope, and explicit configuration identify the applicable instruction family. Operand layouts then establish legality and determine how the logical tile maps to legal instruction shapes. We first introduce the placement-preserving layout operations used during dispatch (\S\ref{sec:layout-xforms}), then illustrate family-specific lowering with TMA copy and Trainium GEMM (\S\ref{sec:tma-lowering}, \S\ref{sec:trn-gemm}). \S\ref{sec:dispatch-contract} presents the complete dispatch contract and the controls exposed to programmers.

\subsection{Dispatch-Oriented Layout Operations}
\label{sec:layout-xforms}

\textbf{Canonicalize.}
Because the same tensor can be written with different iter decompositions, offsets, or replica forms, the compiler needs to answer a basic question: do two structurally different layouts actually mean the same tensor placement?
Dispatch becomes much simpler if we first normalize such layouts into a common form.

For example,
the two layouts
\[
\begin{pmatrix}
  2 & 4 \\
  4@\texttt{m} & 1@\texttt{m}
\end{pmatrix}
\qquad\text{and}\qquad
\begin{pmatrix}
  8 \\
  1@\texttt{m}
\end{pmatrix}
\]
describe the same row-major placement of 8 elements in memory, even though their iter structures differ.

\textbf{Group.}
Primitive dispatch is usually phrased in terms of logical tensor ranks such as $(M,N)$, $(M,N,K)$, or a tiled grid.
Grouping recovers the logical view needed by the primitive interface without changing the underlying placement.

For example, suppose a layout is written as
\[
\begin{pmatrix}
  4 & 6 \\
  6@\texttt{m} & 1@\texttt{m}
\end{pmatrix}.
\]
A dispatch rule for a 4D tiled primitive may want to view it as shape $(2,2,3,2)$ or $(2,3,2,2)$.
Only some such views are valid, because grouping must respect the ordered iter structure.
In this case, $(2,2,3,2)$ is valid, while other decompositions may not be.

\textbf{Permute.}
After grouping, a tensor may already have the right logical rank for a primitive, but its logical axis order may still not match the primitive interface. Dispatch should be able to reorder logical axes without changing the underlying physical placement.

For example,
suppose a tensor is currently viewed as a grouped layout of shape $(M,K)$:
\[
L=
\begin{pmatrix}
  M & K \\
  s_M@\texttt{m} & s_K@\texttt{m}
\end{pmatrix}.
\]
A matrix-multiply primitive for operand \(A\) may expect the logical order \((K,M)\) instead. It can simply permute the logical axes and obtain the permuted view
\[
L^{\pi}=
\begin{pmatrix}
  K & M \\
  s_K@\texttt{m} & s_M@\texttt{m}
\end{pmatrix},
\]
where \(\pi\) swaps the two logical axes. The physical placement is unchanged: only the logical interpretation is reordered. This is useful when the same underlying tensor needs to match different primitive signatures or backend-specific operand conventions.

\textbf{Slice.}
Tile primitives usually operate on a subregion of a tensor, not the whole tensor.
Given a tensor layout, a logical shape, and a region of interest, dispatch becomes much simpler if we can derive the layout of that subtensor directly.
This avoids re-deriving placement logic from scratch for every primitive instance.

For example, suppose we have
\[
L=\begin{pmatrix}
  2 & 8 & 3 & 8 \\
  192 & 8 & 64 & 1
\end{pmatrix}
\]
(omitting \(@\texttt{m}\) for simplicity).
For \(\sigma=(16,24)\) and
\(\rho=[0\!:\!8)\times[8\!:\!24)\), the primitive only touches the upper-right submatrix, whose layout is
\[
L[\rho;\sigma]\;=\;
\begin{pmatrix}
  1 & 8 & 2 & 8\\
  192 & 8 & 64 & 1
\end{pmatrix}
\;+\;64.
\]
For \(\sigma=(16,24)\) and
\(\rho=[2\!:\!6)\times[4\!:\!12)\), the sliced subtensor layout becomes
\[
L[\rho;\sigma]\;=\;
\begin{pmatrix}
  4 & 2 & 4\\
  8 & 60 & 1
\end{pmatrix}
\;+\;20.
\]

To construct a slice, we first group the parent layout according to
\(\sigma\) and locate the origin of \(\rho\). Evaluating the parent layout
at this origin determines the new offset. Iters fully contained in the
region are retained, while an iter crossed by a region boundary is split
to isolate the requested extent. Dimensions fixed by the region are removed.
The resulting layout is accepted when these transformations produce a single
layout whose admitted shape equals the extent of \(\rho\); other regions
produce a slicing failure.

These layout operators form the common mechanism used by family-specific
lowering. After an instruction family is selected, its rule uses slicing,
canonicalization, grouping, and permutation to expose a legal
instruction-compatible view while preserving placement, then emits the
selected intrinsic over the residual extents. The following two sections
instantiate this mechanism for TMA copy and Trainium GEMM.

\subsection{Example 1: TMA Asynchronous Copy}

\begin{figure}[t]
  \centering
  \includegraphics[width=\columnwidth]{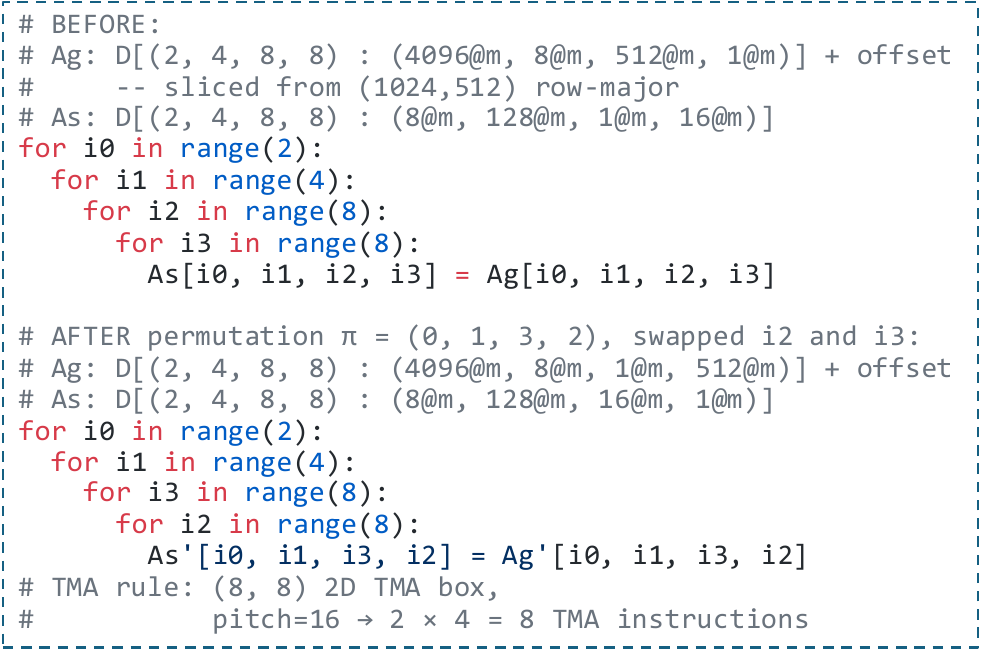}
  \Description{Diagram showing how grouping and permutation help discover a larger TMA box.}
  \caption{Layout grouping and permutation expose a larger legal box within the selected TMA family, reducing the number of issued instructions.}
  \label{fig:example-tma}
\end{figure}

\label{sec:tma-lowering}
TMA allows users to specify a logical multi-dimensional box region in global memory, encoded through \code{cuTensorMap}, and copy it into a contiguous shared-memory region given a starting pointer. Ideally, one issued instruction should copy as much data as possible, thereby reducing the number of TMA instructions in total. Within the selected TMA family, the lowering rule has two tasks: it must identify the largest possible contiguous region in the destination SMEM layout, and it must construct the corresponding \code{cuTensorMap}, which requires the global tensor's shape and stride together with the box extents along each dimension.

\autoref{fig:example-tma} illustrates the core idea of the dispatch algorithm. We first slice and group the SMEM and GMEM layouts so that the analysis is restricted to the region actually being copied. A key observation is that applying the same permutation to both the copy iteration order and the tensor layout preserves copy semantics: each physical location in the source region still maps to the same physical location in the destination.

The TMA rule therefore searches for contiguous iters in the sharding iters of the SMEM layout and attempts to permute them together while maintaining descending stride order. After this transformation, the sharding iters of the permuted GMEM layout naturally determine the global strides and the TMA box extents. In \autoref{fig:example-tma}, the global strides are \((4096, 8, 1, 512)\) and the box extents are \((1, 1, 8, 8)\).

The global shape is also needed by the TMA engine to determine whether a coordinate in the box is logically out of bounds (OOB). For example, consider a tensor with shape \((16,16)\) and layout
$
\begin{pmatrix}
  16 & 16 \\
  16 & 1
\end{pmatrix}.
$
Although coordinate \([0,30]\) maps to physical location \(30\), which lies within the physical span of the tensor, it is logically out of bounds and must be treated as an OOB access. When OOB handling is enabled, the TMA rule preserves the original shape/layout structure while identifying the contiguous SMEM region, keeping the logical global bounds available to descriptor construction.

For the given execution scope, the dispatcher validates that the team kind is compatible with the involved SMEM tensors and uses the active thread set to select the issuing thread.

\subsection{Example 2: AI Accelerator Systolic Array GEMM}

\label{sec:trn-gemm}

We use Trainium~1 as a concrete example. Its matmul unit (the Tensor Engine) accepts only tile layouts with a specific tag pattern. Each layout axis is either a \texttt{P} axis, denoting a 128-way partition axis mapped to the hardware lanes, or an \texttt{F} axis, traversed sequentially. A legal instruction tile must satisfy the following pattern: on the LHS, $K$ must lie on a \texttt{P} axis and $M$ on an \texttt{F} axis; on the RHS, $K$ must again lie on a \texttt{P} axis and $N$ on an \texttt{F} axis; on the output, $M$ must lie on a \texttt{P} axis and $N$ on an \texttt{F} axis. Dispatch therefore reduces to a constrained matching problem: given the operand layouts, find the largest instruction tile whose $M$, $N$, and $K$ dimensions satisfy this tag pattern, and move all remaining extents outside the instruction as loop nests.

The dispatcher proceeds in four steps. First, it regroups each operand so that its logical $M$, $N$, and $K$ dimensions can be identified even when they are split across several layout axes. Second, it chooses the instruction's $K$ dimension by finding a matching \texttt{P} axis in the regrouped $K$ dimensions of both $A$ and $B$. Third, it chooses the instruction's $M$ and $N$ dimensions by finding the largest legal pairings required by the hardware tag pattern. If no valid $M$ pairing exists with $A$ as the LHS, the dispatcher retries with swapped operand roles. Finally, any layout extent not consumed by the instruction tile becomes an outer loop along the corresponding $M$, $N$, or $K$ dimension.

Consider a GEMM with $M{=}512$, $N{=}256$, and $K{=}1024$, with layouts
{
\begin{align*}
L_A &= \begin{pmatrix}
  4 & 128 & 8 & 128 \\
  1024@\texttt{F} & 1@\texttt{F} & 128@\texttt{F} & 1@\texttt{P}
\end{pmatrix}, \\
L_B &= \begin{pmatrix}
  8 & 128 & 2 & 128 \\
  256@\texttt{F} & 1@\texttt{P} & 128@\texttt{F} & 1@\texttt{F}
\end{pmatrix}, \\
L_C &= \begin{pmatrix}
  4 & 128 & 2 & 128 \\
  256@\texttt{F} & 1@\texttt{F} & 128@\texttt{F} & 1@\texttt{P}
\end{pmatrix}.
\end{align*}
}

After regrouping, $A$'s first two axes form its logical $M$ dimension and the last two form its $K$ dimension; for $B$, the first two axes form $K$ and the last two form $N$; for $C$, the first two axes form $M$ and the last two form $N$. Both $A$'s $K$ side and $B$'s $K$ side contain a length-128 \texttt{P} axis, so the dispatcher selects $K{=}128$.

Trying $A$ as the LHS then fails on the $M$ side: the hardware requires an \texttt{F} axis from the LHS paired with a \texttt{P} axis from $C$, but $C$'s $M$ side contains only \texttt{F} axes. The dispatcher therefore swaps operand roles and retries with $B$ as the LHS. Under this assignment, the hardware's $M$ maps to $C$'s $N$ dimension, which does contain a \texttt{P} axis, and the hardware's $N$ maps to $C$'s $M$ dimension, which is entirely \texttt{F}. Both pairings yield extent 128, so the dispatcher selects instruction shape $M{=}N{=}K{=}128$.

The remaining outer extents---$A$'s and $C$'s outer $M$ extent $4$, $B$'s and $C$'s outer $N$ extent $2$, and the shared outer $K$ extent $8$---are emitted as an outer loop nest around the hardware matmul instruction.

\subsection{Dispatch Contract and User Model}
\label{sec:dispatch-contract}

\textbf{Dispatch key and family selection.} \sysname{} uses a registry-based and deterministic dispatch policy. A tile primitive call forms a dispatch key from the primitive kind, target backend, operand storage roles, current execution scope, and explicit \code{config}. Configuration can directly select an instruction family when the programming interface requires an explicit hardware protocol. For the remaining keys, the registry provides at most one default specialized family. Keys without a registered implementation produce a compilation error.

After resolving the instruction family, the corresponding rule checks its hardware contract, including relevant dtype, storage roles, alignment, operand layouts, logical shape, execution scope, and family-specific requirements such as descriptor or out-of-bounds expressibility. The rule uses the placement-preserving layout operations from \S\ref{sec:layout-xforms} to expose an instruction-compatible view. \code{config.shape} selects a requested legal instruction shape when specified. Otherwise, the dispatcher chooses the largest legal instruction shape in the family. Code generation emits residual loops around the selected intrinsic, fusing compatible loops, hoisting invariant descriptors, and simplifying offsets within the primitive call.

\begin{table}[t]
\centering
\small
\caption{Representative tile primitives, specialized instruction families,
and registered fallbacks.}
\label{tab:dispatch-registry}
\setlength{\tabcolsep}{3pt}
\settowidth{\tabletextwidth}{\code{copy_async}}
\setlength{\tabletextwidth}{\dimexpr(\columnwidth-\tabletextwidth-4\tabcolsep)/2\relax}
\begin{tabular}{@{}l>{\raggedright\arraybackslash}p{\tabletextwidth}
                    >{\raggedright\arraybackslash}p{\tabletextwidth}@{}}
\toprule
Primitive & Specialized family & Fallback \\
\midrule
\code{copy}
  & \code{ldmatrix/stmatrix}, NKI copy
  & vector/scalar load/store \\

\code{copy_async}
  & TMA, \code{tcgen05.cp/ld/st}, \code{cp.async}
  & none \\

\code{gemm}
  & \code{mma.sync.aligned}, \code{nki.matmul}
  & scalar FMA \\

\code{gemm_async}
  & \code{tcgen05.mma}
  & none \\

\code{reduce}
  & warp shuffle, NKI tensor reduce
  & vector/scalar arithmetic \\

\code{pointwise}
  & NKI/vector arithmetic
  & scalar arithmetic \\

\code{select}
  & \code{nki.affine_select}
  & scalar \\

\code{composite}
  & NKI composite operations
  & scalar \\
\bottomrule
\end{tabular}
\end{table}

\textbf{Explicit asynchronous families.} Asynchronous instruction families expose synchronization behavior to the surrounding kernel schedule. TMA and \code{cp.async}, for example, use different issue, wait, and barrier protocols, so their family selection is encoded explicitly in the primitive configuration. Similarly, \code{tcgen05.mma} is exposed as \code{gemm_async}, while scalar FMA is available through \code{gemm}. These source-level distinctions make the synchronization protocol and primitive semantics visible at each call site.

\textbf{Locality across primitive calls.} Each tile primitive is dispatched independently from the explicit state at its call site. Storage allocation, staging, synchronization, and pipeline structure remain part of the surrounding native-level schedule. A global-to-register transfer that uses shared-memory staging, for example, is written as a global-to-SMEM primitive followed by an SMEM-to-register primitive. Each call is then dispatched using its own storage roles, scope, layouts, and configuration. Tensor declarations and library constructors attach explicit layouts to
program values, and these layouts carry placement decisions between primitive calls. Dispatch preserves the logical tile and its placement, so selecting another legal instruction shape changes the instruction loop generated for the current primitive while leaving the declared placement available to subsequent calls.

\textbf{User control and override.} Programmers can steer dispatch at two granularities. The tensor slice passed to a primitive defines its logical operand tile and therefore its working set. Within that tile, \code{config} can select an instruction family or a legal instruction shape. The largest-legal-shape policy serves as the default shape-selection rule for a family. A programmer optimizing cache residency, register pressure, occupancy, or instruction scheduling can change the logical tile or select another legal instruction shape. When a specialized implementation cannot satisfy its hardware constraints, the registry may provide a semantically equivalent fallback. Dispatch uses that fallback when available and reports an error for unsupported cases. Together, logical-tile selection, family selection, shape selection, and registered fallbacks expose the main dispatch choices directly to the programmer.

\textbf{Predictability and diagnostics.} The registry and its legality rules form a public dispatch contract. Compiler diagnostics report the selected rule and instruction shape. When dispatch reaches a registered fallback, the diagnostic identifies the specialized rule's failed legality constraints. Programmers can inspect the lowered IR and generated backend code and can select an instruction family or shape through \code{config}. Peak-performance programming still relies on knowledge of hardware constraints such as scope, layout, alignment, and instruction shape. \sysname{} exposes these constraints through a stable tile-primitive interface and makes the resulting dispatch decision inspectable.

\section{Evaluation}
\label{sec:evaluation}

We implement the \sysname{} compiler on top of TensorIR~\cite{feng2023tensorir} in Apache TVM~\cite{chen2018tvm}. We use Triton/Gluon (3.6.0), Triton-dist (0.0.2-rc), CUDA (13.0), DeepGEMM (commit 11cbe26), CuTeDSL (4.4.2), FlashAttention (2.8.3) and Trainium stacks (2.26.0).
The same design can also be applied to other ML compilers and DSLs.
This section asks the following questions:

\begin{itemize}
    \item Can \sysname{} deliver near-best performance on the latest GPU architecture~(\S\ref{sec:single-gpu}, \S\ref{sec:multi-gpu})?
    \item Can \sysname{} support heterogeneous hardware environments~(\S\ref{sec:ai-acc})?
\end{itemize}

Finally, based on these results, we examine which low-level implementation tasks \sysname{}'s design elements automate while preserving explicit scheduling control (\S\ref{sec:design-elem}).

\begin{figure*}[!t]
  \centering
  \includegraphics[width=\textwidth]{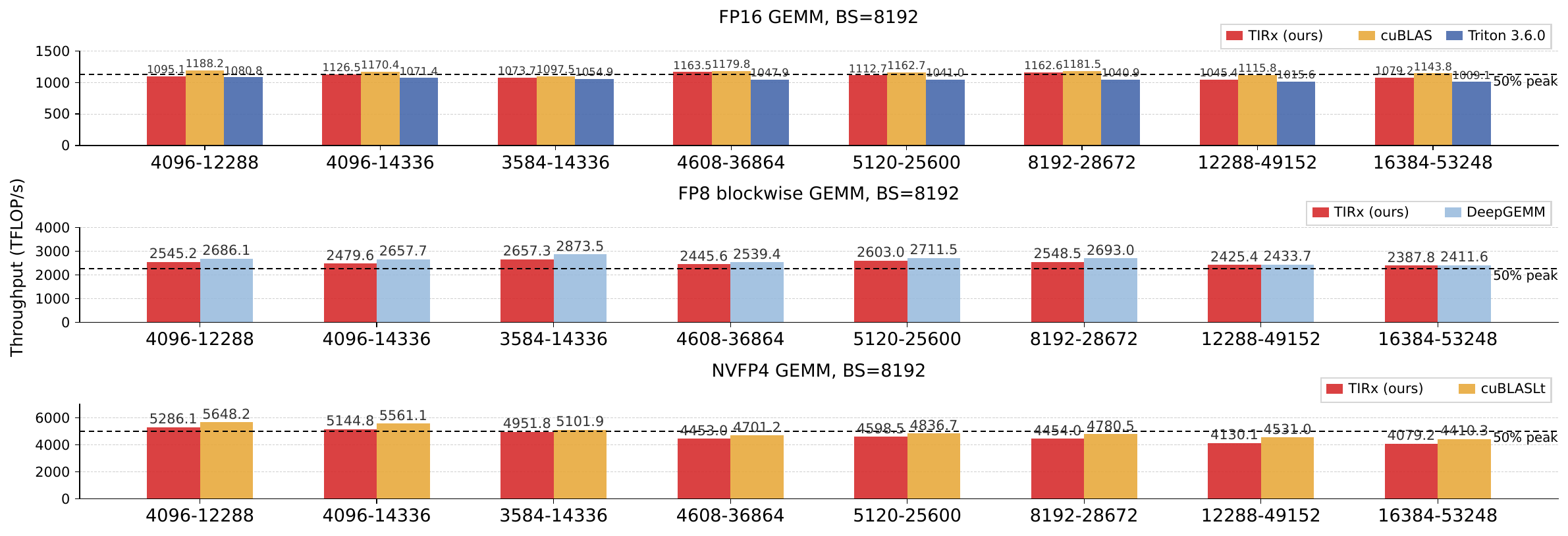}
\Description{Bar charts comparing FP16, FP8 block-scaled, and NVFP4 GEMM
throughput for \sysname{} against cuBLAS, Triton, DeepGEMM, and cuBLASLt across
different weight shapes.}
  \caption{FP16 GEMM and FP8, NVFP4 block-scaled GEMM throughput (TFLOP/s) at batch size 8192 across different weight shapes. The dashed line marks 50\% of device peak (for FP16, FP8 and NVFP4, respectively). \sysname{} is within 4\% of the vendor-tuned cuBLAS on average for FP16 GEMMs while significantly outperforming Triton. Each $x$-axis label denotes the weight shape $N$-$K$ (with $M{=}8192$). Higher is better.}
  \label{fig:eval-op}
\end{figure*}

\subsection{Kernel Performance on NVIDIA B200}

\label{sec:single-gpu}
In this section, we evaluate the performance of kernels generated by the \sysname{} compiler on the latest NVIDIA B200 GPU.
We conduct the evaluation on a DGX B200 server with CUDA~13.0 and driver~580.82.07.
Each experiment runs 1,000 warm-up iterations; we report throughput computed from the average runtime over 3,000 repeat iterations.

\textbf{FP16 GEMM (\autoref{fig:eval-op} top).} We evaluate FP16 GEMMs.
We use batch size 8192 and real-world weight shapes from Qwen3-8B/32B~\cite{yang2025qwen3},
LLaMA-3.1-8B/70B/405B~\cite{grattafiori2024llama}, Gemma-2-9B/27B~\cite{team2024gemma}, and GPT-3-175B~\cite{brown2020language}.
We use cuBLAS~\cite{cublas} and Triton~\cite{tillet2019triton} as our baselines.
For FP16 GEMM, \sysname{} is within 2\% of cuBLAS on several shapes and within 4\% on average, while Triton consistently trails cuBLAS, with most gaps falling around 9--12\%. Note that cuBLAS is a vendor-tuned SASS library; coming within 4\% on average from a Python DSL is a strong result for a compiler-based approach.

\textbf{Low-precision block-scaled GEMM (\autoref{fig:eval-op} mid and bottom).}
We also evaluate low-precision block-scaled GEMMs against the best available hand-tuned baselines. For FP8 (e8m0) block-wise scaling, \sysname{} reaches 0.92--0.96$\times$ of DeepGEMM~\cite{deepgemm} (a hand-optimized CUDA library); for NVFP4 (e2m1), \sysname{} reaches 0.92--0.97$\times$ of cuBLASLt~\cite{cublas}.

\textbf{FlashAttention 4 (\autoref{fig:eval-fa4}).}
We further evaluate forward multi-head attention (MHA) on B200 with head dimension 128 under both non-causal and causal settings. In the non-causal case, \sysname{} outperforms Triton/Gluon by \(1.16\times\) to \(1.34\times\) (averaging \(1.23\times\)) and matches FlashAttention-4 and cuDNN at \(1.00\times\) and \(0.99\times\) on average, respectively. In the causal case, \sysname{} exceeds Triton/Gluon from 2K onward by \(1.07\times\) to \(1.28\times\) (averaging \(1.18\times\)) and reaches about \(0.96\times\) of both FlashAttention-4 and cuDNN on average.

\sysname{} implements the optimization techniques of FlashAttention~4~\cite{zadouri2026flashattention} in about 800 lines of Python, comparable to Gluon. TMA copy, TMEM load/store, \code{tcgen05} GEMM, and warpgroup register operators are abstracted as tile primitives. Fine-grained optimizations such as exponential emulation via polynomial approximation, which require carefully arranged per-thread register operations, are directly programmed at the native level.

\begin{figure*}[!t]
  \centering
  \includegraphics[width=\textwidth]{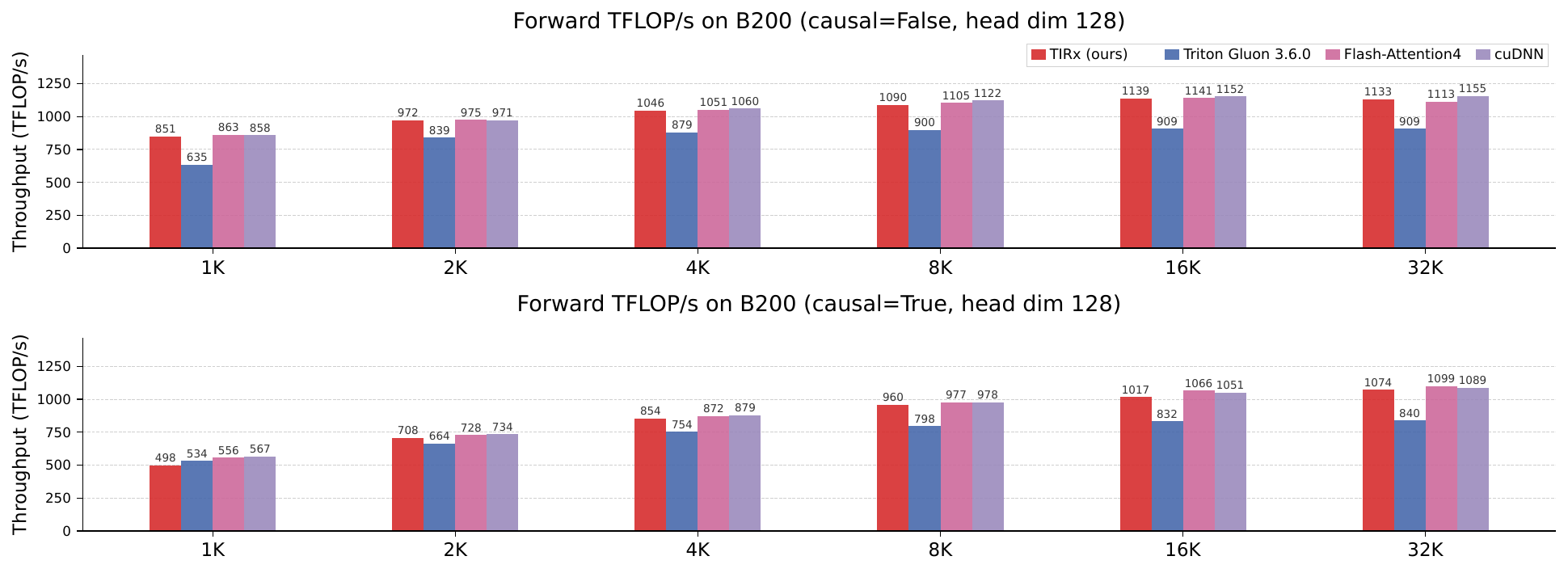}
  \Description{Bar charts comparing forward MHA throughput on B200 for \sysname{}, FlashAttention-4, cuDNN, and Triton/Gluon across sequence lengths.}
  \caption{Forward multi-head attention (MHA) throughput on B200 (head dimension 128) for non-causal and causal settings across sequence lengths from 1K to 32K. Higher is better.}
  \label{fig:eval-fa4}
\end{figure*}

\subsection{Multi-GPU Kernel Performance}

\label{sec:multi-gpu}
\begin{figure*}[!t]
  \centering
  \includegraphics[width=\textwidth]{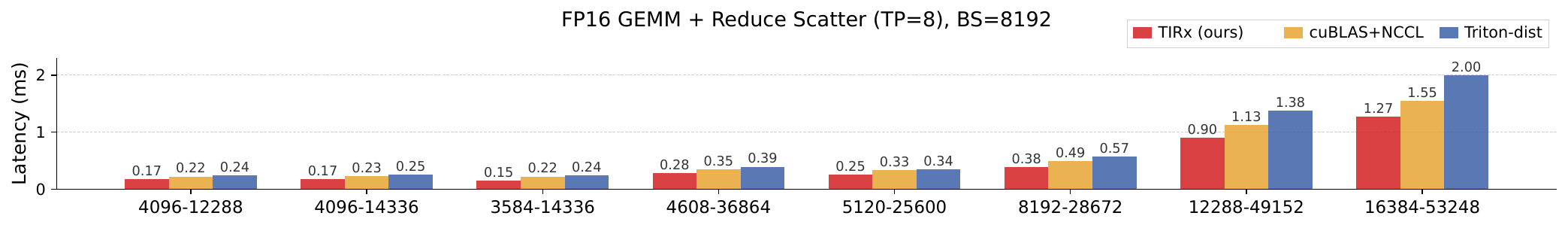}
  \Description{Bar chart comparing FP16 GEMM plus Reduce-Scatter latency for \sysname{}, cuBLAS plus NCCL, and Triton-distributed across weight shapes.}
  \caption{FP16 GEMM + Reduce-Scatter latency (ms) evaluated across different weight shapes. Lower is better.}
  \label{fig:multi-gpu}
\end{figure*}

This section evaluates \sysname{}’s ability to generate distributed kernels that overlap communication and computation.
We choose GEMM + Reduce-Scatter workloads from the same FP16 GEMM config in the previous subsection.
The \sysname{} kernel composes a distributed tensor, invokes the sum primitive which is then dispatched to \code{multimem.ld_reduce} on B200.
We pick cuBLAS + NCCL~\cite{nccl} and Triton-distributed~\cite{zheng2025tritondistributedprogrammingoverlappingkernels} as our baselines
and run the evaluation on a DGX B200 server.
The results are shown in \autoref{fig:multi-gpu}. \sysname{} delivers the lowest latency across the cases, gaining up to
\(1.40\times\) speedup over the best baseline.
The speedup comes from a fine-grained overlap of communication and computation in a single kernel explicitly programmed through \sysname{}'s native-level access, resulting in better memory bandwidth and Tensor Core utilization.
Triton-distributed is the strongest distributed compiler baseline we found.
In the evaluated release, its Blackwell GEMM support is still maturing, and
its lower B200 GEMM throughput contributes to the end-to-end latency gap in
this experiment.

\subsection{Supporting Heterogeneous Hardware Backends}

This section evaluates \sysname{}’s ability to target heterogeneous backends.
We evaluate FP16 GEMM and multi-head attention (MHA) performance on a \texttt{trn1.2xlarge} AWS instance with Trainium1 AI accelerator.
We compare \sysname{} kernels to vendor-provided reference libraries (handcrafted in Neuron Kernel Interface (NKI) DSL~\cite{nki}) and the Neuron compiler~\cite{neuron},
and report performance relative to the Neuron compiler.
The results are shown in \autoref{fig:trn}. Our FP16 GEMM kernel ($M=N=K$) matches the performance of the handcrafted NKI library in every configuration.
On the MHA workload, the \sysname{} kernel achieves up to 1.44$\times$ speedup and 1.26$\times$ on average over NKI.
\sysname{} obtains the speedup by orchestrating the software pipeline schedule and memory allocation plan.
The manually optimized NKI implementation takes 120 lines for GEMM and 1188 lines for MHA, while the \sysname{} kernel uses only 78 lines for GEMM and 228 lines for MHA.
The \sysname{} DSL simplifies primitive scheduling and address calculation, enabling us to generate efficient NKI programs from a higher-level description.

\label{sec:ai-acc}
\begin{figure}[t]
  \centering
  \includegraphics[width=\columnwidth]{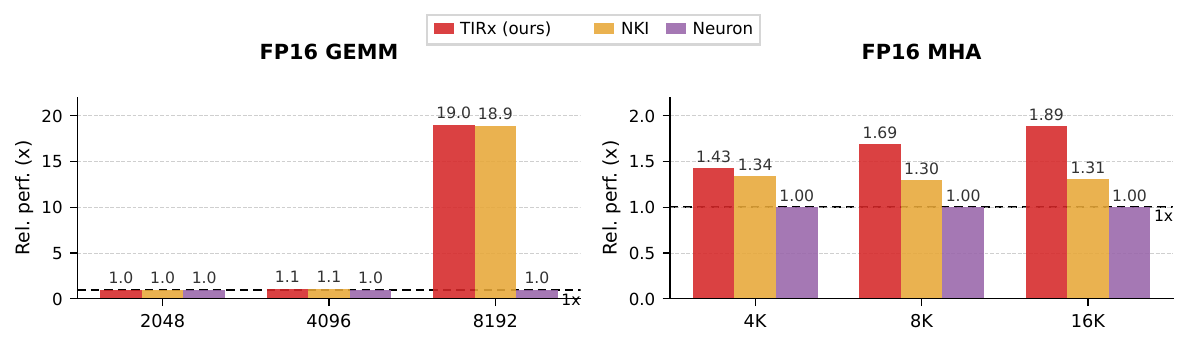}
  \Description{Bar chart comparing FP16 GEMM and MHA performance on Trainium1 for \sysname{}, NKI, and Neuron compiler.}
  \caption{FP16 GEMM and multi-head attention (MHA) test results. FP16 GEMM is tested on square shapes. MHA is tested with varying input lengths with no causal mask.}
  \label{fig:trn}
\end{figure}

\subsection{Discussion: What Do \sysname{}'s Design Elements Automate?}

The evaluation above illustrates how \sysname{} divides work between
programmers and the compiler. Peak-performance kernels retain explicit
scheduling, synchronization, pipelines, and hardware orchestration, while
scope, layout, and dispatch automate recurring address calculation, operand
preparation, and instruction-loop generation. We use LOC as evidence
of low-level implementation boilerplate reduction.

For FP8 block-scaled GEMM, DeepGEMM uses roughly 550 lines of CUDA C++
versus about 250 lines of Python in \sysname{}. The native implementation
still performs address calculation and writes instruction-loop structure
explicitly, while \sysname{} derives addresses from layouts and captures
repeated instruction applications in tile primitives.

For FP16 GEMM, the evaluated Triton kernel is about 80 lines; its generated
schedule uses two load warps, one GEMM warp, one write-back warp, and no
cluster cooperation. The roughly 180-line \sysname{} kernel explicitly
assigns one load warp, two GEMM warps, two write-back warps, and a size-2
cluster, providing context for the throughput gap in
\autoref{fig:eval-op}. This limitation has motivated concurrent work such as Gluon~\cite{gluon}, which exposes more warp/layout control
while retaining the underlying power-of-two layout constraints.

Defaults and helpers supply many \sysname{} layouts. The only
user-declared layout in FA\mbox{-}4 is \code{wg_rf_layout(cols)}. It supplies
the register fragment for TMEM/register transfers. Performance comes from
the warp-specialized schedule, asynchronous pipeline, and native-level
exponential emulation. The conditional TMEM-to-register load uses this layout
plus one copy primitive; CuTeDSL requires atom definition, layout composition,
slicing, source/destination partitioning, and the final copy.

\label{sec:design-elem}

\section{Related Work}
\label{sec:related}

\textbf{Layout Systems.} There are several lines of work formalizing the mapping from data tensors to hardware units~\cite{Graphene, ding2025tilus}, with most efforts addressing only part of the stack. The closest works to ours are CuTe~\cite{cute} and Triton linear layouts~\cite{linear}.

\emph{Relation to CuTe.} \sysname{} uses the same shape and stride arithmetic as CuTe. CuTe generalizes strides to elements of an integer module and commonly uses unit vectors to target multi-dimensional TMA coordinates in global memory. \sysname{} introduces explicitly named axes to form a vocabulary of hardware resources. Inside an atom, CuTe maps the pair \((t,v)\) to a logical index for work partitioning and remains single-valued. \sysname{} maps logical indices to physical coordinates and adds \(R\) and \(O\) for replication and offset, making the forward map set-valued.

\emph{Relation to linear layouts.} Linear layouts employ bit-linear forms for layout conversion and swizzle compatibility.
This design enforces power-of-two shapes in the internal layout, which can be limiting for cases such as DeepGEMM~\cite{deepgemm}
and distributed settings.
For a tensor of shape $24\times24$ with column‑major layout
$
f(i) = \lfloor i /24 \rfloor + (i \%24)*24,
$
we have $f(1)=24$, $f(2)=48$, $f(1\ \text{XOR}\ 2)=f(3)=72$, while $f(1)\ \text{XOR}\ f(2)=40\neq f(1\ \text{XOR}\ 2)$, so no bit linear $f$ over $\mathcal{F}_2$ can satisfy this.

CuTe and linear layouts target intra-GPU layouts, while \sysname{} also supports heterogeneous backends.
\sysname{} complements and can interoperate with these systems.
It generalizes the shape-and-stride mechanism of common array APIs~\cite{numpy, paszke2019pytorch} with named axes to support device memory and heterogeneous platforms.

\textbf{Deep Learning Compilers and DSLs.}
Halide and TVM separate algorithm and schedule~\cite{halide, chen2018tvm, feng2023tensorir}. CuTeDSL exposes atoms that mirror hardware instructions and gives users loop transformations and thread-level partitioning~\cite{CuTeDSL}. Graphene introduces a GPU-centric intermediate representation for optimized tensor computations, targeting intra-GPU kernel optimization~\cite{Graphene}. Triton provides a block-collective programming model and lets the compiler decide per-thread implementations~\cite{tillet2019triton}. Gluon~\cite{gluon} further extends Triton to support more explicit control over tensor layouts and warp-level operations. TileLang and Tilus extend this style while keeping tile-level abstractions for specific scenarios~\cite{wang2025tilelang, ding2025tilus}. Pallas offers a low-level kernel DSL integrated with JAX and TPU backends, and AWS Neuron provides a Trainium stack~\cite{pallas, nki}.
Our techniques can potentially be incorporated into those efforts to broaden coverage and improve productivity.

\textbf{Distributed Machine Learning Frameworks.}
Many systems study sharding and placement over device meshes. Mesh TensorFlow introduced named dimensions for SPMD. JAX GSPMD and shard map unify data and model parallelism with PartitionSpec, while Alpa and FlexFlow search over parallelization choices~\cite{mesh-tf, xu2021gspmd, alpa, flexflow}. TensorFlow DTensor and PyTorch Distributed Tensor surface sharding and replication in the core frameworks~\cite{tensorflow, paszke2019pytorch}. TileLink and Triton Distributed bring collectives into kernels so communication can overlap execution at fine granularity~\cite{zheng2025tilelink, zheng2025tritondistributedprogrammingoverlappingkernels}.
\sysname{}'s layout can be used to cover distributed tensor formats while adding intra-GPU tiling details that are not captured by current distributed formats.

\section{Conclusion}

We presented \textbf{\sysname{}}, a native-first compiler for ML kernels across GPUs and AI accelerators. \sysname{} organizes kernel compilation around three explicit interfaces: execution scope, tensor layout with \textbf{D} (shard), \textbf{R} (replica), and \textbf{O} (offset) semantics, and scope- and layout-aware tile primitive dispatch. By adding only a thin dispatch layer above native-level code, \sysname{} delivers competitive performance and practical wins across modern GPUs, multi-GPU kernels, and AI accelerators, offering a solid foundation for reusable tile primitives across a heterogeneous software--hardware stack.

\begin{acks}
This work is supported in part by the \grantsponsor{NSF}{National Science Foundation}{https://www.nsf.gov/} under grant number \grantnum{NSF}{IIS2543593} and by gifts from NVIDIA and Google. We also thank NVIDIA for DGX B200 support.

ChatGPT and Claude Code were used to assist with drafting and editing portions of the manuscript and with engineering development.
\end{acks}

\bibliographystyle{ACM-Reference-Format}
\bibliography{references}

@misc{linear,
      title={Linear Layouts: Robust Code Generation of Efficient Tensor Computation Using $\mathbb{F}_2$}, 
      author={Keren Zhou and Mario Lezcano and Adam Goucher and Akhmed Rakhmati and Jeff Niu and Justin Lebar and Pawel Szczerbuk and Peter Bell and Phil Tillet and Thomas Raoux and Zahi Moudallal},
      year={2025},
      eprint={2505.23819},
      archivePrefix={arXiv},
      primaryClass={cs.PL},
      url={https://arxiv.org/abs/2505.23819}, 
}

@article{guo2025sonicmoe,
  title={SonicMoE: Accelerating MoE with IO and Tile-aware Optimizations},
  author={Guo, Wentao and Mishra, Mayank and Cheng, Xinle and Stoica, Ion and Dao, Tri},
  journal={arXiv preprint arXiv:2512.14080},
  year={2025}
}

@article{zadouri2026flashattention,
  title={FlashAttention-4: Algorithm and Kernel Pipelining Co-Design for Asymmetric Hardware Scaling},
  author={Zadouri, Ted and Hoehnerbach, Markus and Shah, Jay and Liu, Timmy and Thakkar, Vijay and Dao, Tri},
  journal={arXiv preprint arXiv:2603.05451},
  year={2026}
}

@article{ye2025flashinfer,
  title={Flashinfer: Efficient and customizable attention engine for llm inference serving},
  author={Ye, Zihao and Chen, Lequn and Lai, Ruihang and Lin, Wuwei and Zhang, Yineng and Wang, Stephanie and Chen, Tianqi and Kasikci, Baris and Grover, Vinod and Krishnamurthy, Arvind and others},
  journal={Proceedings of Machine Learning and Systems},
  volume={7},
  year={2025}
}

@article{shah2024flashattention,
  title={Flashattention-3: Fast and accurate attention with asynchrony and low-precision},
  author={Shah, Jay and Bikshandi, Ganesh and Zhang, Ying and Thakkar, Vijay and Ramani, Pradeep and Dao, Tri},
  journal={Advances in Neural Information Processing Systems},
  volume={37},
  pages={68658--68685},
  year={2024}
}

@article{mesh-tf,
  title={Mesh-tensorflow: Deep learning for supercomputers},
  author={Shazeer, Noam and Cheng, Youlong and Parmar, Niki and Tran, Dustin and Vaswani, Ashish and Koanantakool, Penporn and Hawkins, Peter and Lee, HyoukJoong and Hong, Mingsheng and Young, Cliff and others},
  journal={Advances in neural information processing systems},
  volume={31},
  year={2018}
}

@misc{neuron,
  author       = {{Amazon Web Services}},
  title        = {About the AWS Neuron SDK},
  howpublished = {\emph{AWS Neuron Documentation}},
  year         = {2025},
  note         = {Accessed Oct~29,~2025.},
  url          = {https://awsdocs-neuron.readthedocs-hosted.com/en/latest/about-neuron/index.html}
}

@misc{ptdtensor,
  author       = {{PyTorch Contributors}},
  title        = {torch.distributed.tensor: DTensor Class APIs},
  howpublished = {\emph{PyTorch 2.9 Documentation}},
  year         = {2025},
  note         = {Created Jun~13,~2025; Last updated Aug~23,~2025. Accessed Oct~30,~2025.},
  url          = {https://docs.pytorch.org/docs/stable/distributed.tensor.html#dtensor-class-apis}
}

@misc{gluon,
  author       = {{Triton Developers}},
  title        = {Gluon tutorial: 01-intro.py},
  howpublished = {\emph{GitHub repository}},
  year         = {2025},
  note         = {Accessed Oct~29,~2025. Path: python/tutorials/gluon/01-intro.py},
  url          = {https://github.com/triton-lang/triton/blob/main/python/tutorials/gluon/01-intro.py}
}

@inproceedings {tensorflow,
author = {Mart{\'\i}n Abadi and Paul Barham and Jianmin Chen and Zhifeng Chen and Andy Davis and Jeffrey Dean and Matthieu Devin and Sanjay Ghemawat and Geoffrey Irving and Michael Isard and Manjunath Kudlur and Josh Levenberg and Rajat Monga and Sherry Moore and Derek G. Murray and Benoit Steiner and Paul Tucker and Vijay Vasudevan and Pete Warden and Martin Wicke and Yuan Yu and Xiaoqiang Zheng},
title = {{TensorFlow}: A System for {Large-Scale} Machine Learning},
booktitle = {12th USENIX Symposium on Operating Systems Design and Implementation (OSDI 16)},
year = {2016},
isbn = {978-1-931971-33-1},
address = {Savannah, GA},
pages = {265--283},
url = {https://www.usenix.org/conference/osdi16/technical-sessions/presentation/abadi},
publisher = {USENIX Association},
month = nov
}

@inproceedings{halide,
author = {Ragan-Kelley, Jonathan and Barnes, Connelly and Adams, Andrew and Paris, Sylvain and Durand, Fr\'{e}do and Amarasinghe, Saman},
title = {Halide: a language and compiler for optimizing parallelism, locality, and recomputation in image processing pipelines},
year = {2013},
isbn = {9781450320146},
publisher = {Association for Computing Machinery},
address = {New York, NY, USA},
url = {https://doi.org/10.1145/2491956.2462176},
doi = {10.1145/2491956.2462176},
booktitle = {Proceedings of the 34th ACM SIGPLAN Conference on Programming Language Design and Implementation},
pages = {519–530},
numpages = {12},
location = {Seattle, Washington, USA},
series = {PLDI '13}
}

@article{yang2025qwen3,
  title={Qwen3 technical report},
  author={Yang, An and Li, Anfeng and Yang, Baosong and Zhang, Beichen and Hui, Binyuan and Zheng, Bo and Yu, Bowen and Gao, Chang and Huang, Chengen and Lv, Chenxu and others},
  journal={arXiv preprint arXiv:2505.09388},
  year={2025}
}

@misc{nccl,
  author       = {{NVIDIA Corporation}},
  title        = {Overview of NCCL},
  howpublished = {\emph{NCCL 2.28.6 Documentation}},
  year         = {2025},
  note         = {Accessed Oct~28,~2025.},
  url          = {https://docs.nvidia.com/deeplearning/nccl/user-guide/docs/overview.html}
}

@article{grattafiori2024llama,
  title={The llama 3 herd of models},
  author={Grattafiori, Aaron and Dubey, Abhimanyu and Jauhri, Abhinav and Pandey, Abhinav and Kadian, Abhishek and Al-Dahle, Ahmad and Letman, Aiesha and Mathur, Akhil and Schelten, Alan and Vaughan, Alex and others},
  journal={arXiv preprint arXiv:2407.21783},
  year={2024}
}

@article{team2024gemma,
  title={Gemma 2: Improving open language models at a practical size},
  author={Team, Gemma and Riviere, Morgane and Pathak, Shreya and Sessa, Pier Giuseppe and Hardin, Cassidy and Bhupatiraju, Surya and Hussenot, L{\'e}onard and Mesnard, Thomas and Shahriari, Bobak and Ram{\'e}, Alexandre and others},
  journal={arXiv preprint arXiv:2408.00118},
  year={2024}
}

@article{brown2020language,
  title={Language models are few-shot learners},
  author={Brown, Tom and Mann, Benjamin and Ryder, Nick and Subbiah, Melanie and Kaplan, Jared D and Dhariwal, Prafulla and Neelakantan, Arvind and Shyam, Pranav and Sastry, Girish and Askell, Amanda and others},
  journal={Advances in neural information processing systems},
  volume={33},
  pages={1877--1901},
  year={2020}
}

@article{cecka2026cute,
  title={CuTe Layout Representation and Algebra},
  author={Cecka, Cris},
  journal={arXiv preprint arXiv:2603.02298},
  year={2026}
}

@misc{stanfordThunderKittens20,
	author = {Stuart Sul and Chris R{\'e}},
	title = {{T}hunder{K}ittens 2.0: {E}ven {F}aster {K}ernels for {Y}our {G}{P}{U}s --- hazyresearch.stanford.edu},
	howpublished = {\url{https://hazyresearch.stanford.edu/blog/2026-02-19-tk-2}},
	year = {},
	note = {[Accessed 09-04-2026]},
}

@article{cheng2025mirage,
  title={Mirage Persistent Kernel: A Compiler and Runtime for Mega-Kernelizing Tensor Programs},
  author={Cheng, Xinhao and Zhang, Zhihao and Zhou, Yu and Ji, Jianan and Jiang, Jinchen and Zhao, Zepeng and Xiao, Ziruo and Ye, Zihao and Huang, Yingyi and Lai, Ruihang and others},
  journal={arXiv preprint arXiv:2512.22219},
  year={2025}
}

@misc{cublas,
  author       = {{NVIDIA Corporation}},
  title        = {cuBLAS},
  howpublished = {\emph{CUDA Toolkit Documentation}},
  year         = {2025},
  note         = {v13.0. Last updated Oct~02,~2025. Accessed Oct~28,~2025.},
  url          = {https://docs.nvidia.com/cuda/cublas/index.html}
}

@misc{deepgemm,
  author       = {{DeepSeek}},
  title        = {{DeepGEMM}},
  howpublished = {\emph{GitHub repository}},
  year         = {2025},
  note         = {Version v2.1.1.post3 (released Oct~15,~2025). Accessed Oct~28,~2025.},
  url          = {https://github.com/deepseek-ai/DeepGEMM}
}

@inproceedings{feng2023tensorir,
  title={Tensorir: An abstraction for automatic tensorized program optimization},
  author={Feng, Siyuan and Hou, Bohan and Jin, Hongyi and Lin, Wuwei and Shao, Junru and Lai, Ruihang and Ye, Zihao and Zheng, Lianmin and Yu, Cody Hao and Yu, Yong and others},
  booktitle={Proceedings of the 28th ACM International Conference on Architectural Support for Programming Languages and Operating Systems, Volume 2},
  pages={804--817},
  year={2023}
}

@inproceedings{chen2018tvm,
  title={$\{$TVM$\}$: An automated $\{$End-to-End$\}$ optimizing compiler for deep learning},
  author={Chen, Tianqi and Moreau, Thierry and Jiang, Ziheng and Zheng, Lianmin and Yan, Eddie and Shen, Haichen and Cowan, Meghan and Wang, Leyuan and Hu, Yuwei and Ceze, Luis and others},
  booktitle={13th USENIX Symposium on Operating Systems Design and Implementation (OSDI 18)},
  pages={578--594},
  year={2018}
}

@article{numpy,
  title={Array programming with NumPy},
  author={Harris, Charles R and Millman, K Jarrod and Van Der Walt, St{\'e}fan J and Gommers, Ralf and Virtanen, Pauli and Cournapeau, David and Wieser, Eric and Taylor, Julian and Berg, Sebastian and Smith, Nathaniel J and others},
  journal={nature},
  volume={585},
  number={7825},
  pages={357--362},
  year={2020},
  publisher={Nature Publishing Group UK London}
}

@article{zheng2025tilelink,
  title={Tilelink: Generating efficient compute-communication overlapping kernels using tile-centric primitives},
  author={Zheng, Size and Fang, Jin and Zheng, Xuegui and Hou, Qi and Bao, Wenlei and Zheng, Ningxin and Jiang, Ziheng and Wang, Dongyang and Ye, Jianxi and Lin, Haibin and others},
  journal={arXiv preprint arXiv:2503.20313},
  year={2025}
}

@article{wang2025tilelang,
  title={TileLang: A Composable Tiled Programming Model for AI Systems},
  author={Wang, Lei and Cheng, Yu and Shi, Yining and Tang, Zhengju and Mo, Zhiwen and Xie, Wenhao and Ma, Lingxiao and Xia, Yuqing and Xue, Jilong and Yang, Fan and others},
  journal={arXiv preprint arXiv:2504.17577},
  year={2025}
}

@misc{nki,
  author       = {{Amazon Web Services}},
  title        = {{NKI API Reference Manual}},
  howpublished = {\emph{AWS Neuron Documentation}},
  year         = {2025},
  note         = {Accessed Oct~28,~2025.},
  url          = {https://awsdocs-neuron.readthedocs-hosted.com/en/latest/nki/api/index.html}
}

@misc{cute,
  author       = {{NVIDIA Corporation}},
  title        = {Getting Started With CuTe},
  howpublished = {\emph{NVIDIA CUTLASS Documentation}},
  year         = {2025},
  note         = {Last updated Sep~24,~2025. Accessed Oct~28,~2025.},
  url          = {https://docs.nvidia.com/cutlass/media/docs/cpp/cute/00_quickstart.html}
}

@misc{pallas,
  author       = {{The JAX Authors}},
  title        = {Pallas: a JAX kernel language},
  howpublished = {\emph{JAX documentation}},
  year         = {2024},
  note         = {Accessed Oct~28,~2025.},
  url          = {https://docs.jax.dev/en/latest/pallas/index.html}
}

@misc{CuTeDSL,
  author       = {{NVIDIA Corporation}},
  title        = {CuTe DSL: Introduction},
  howpublished = {\emph{NVIDIA CUTLASS Documentation}},
  year         = {2025},
  note         = {Last updated Sep~24,~2025. Accessed Oct~28,~2025.},
  url          = {https://docs.nvidia.com/cutlass/media/docs/pythonDSL/cute_dsl_general/dsl_introduction.html}
}

@article{flexflow,
  title={Beyond data and model parallelism for deep neural networks.},
  author={Jia, Zhihao and Zaharia, Matei and Aiken, Alex},
  journal={Proceedings of Machine Learning and Systems},
  volume={1},
  pages={1--13},
  year={2019}
}

@inproceedings {alpa,
author = {Lianmin Zheng and Zhuohan Li and Hao Zhang and Yonghao Zhuang and Zhifeng Chen and Yanping Huang and Yida Wang and Yuanzhong Xu and Danyang Zhuo and Eric P. Xing and Joseph E. Gonzalez and Ion Stoica},
title = {Alpa: Automating Inter- and {Intra-Operator} Parallelism for Distributed Deep Learning},
booktitle = {16th USENIX Symposium on Operating Systems Design and Implementation (OSDI 22)},
year = {2022},
isbn = {978-1-939133-28-1},
address = {Carlsbad, CA},
pages = {559--578},
url = {https://www.usenix.org/conference/osdi22/presentation/zheng-lianmin},
publisher = {USENIX Association},
month = jul
}

@inproceedings{Graphene,
author = {Hagedorn, Bastian and Fan, Bin and Chen, Hanfeng and Cecka, Cris and Garland, Michael and Grover, Vinod},
title = {Graphene: An IR for Optimized Tensor Computations on GPUs},
year = {2023},
isbn = {9781450399180},
publisher = {Association for Computing Machinery},
address = {New York, NY, USA},
url = {https://doi.org/10.1145/3582016.3582018},
doi = {10.1145/3582016.3582018},
booktitle = {Proceedings of the 28th ACM International Conference on Architectural Support for Programming Languages and Operating Systems, Volume 3},
pages = {302–313},
numpages = {12},
location = {Vancouver, BC, Canada},
series = {ASPLOS 2023}
}

@article{ding2025tilus,
  title={Tilus: A Virtual Machine for Arbitrary Low-Precision GPGPU Computation in LLM Serving},
  author={Ding, Yaoyao and Hou, Bohan and Zhang, Xiao and Lin, Allan and Chen, Tianqi and Hao, Cody Yu and Wang, Yida and Pekhimenko, Gennady},
  journal={arXiv preprint arXiv:2504.12984},
  year={2025}
}

@inproceedings{tillet2019triton,
  title={Triton: an intermediate language and compiler for tiled neural network computations},
  author={Tillet, Philippe and Kung, Hsiang-Tsung and Cox, David},
  booktitle={Proceedings of the 3rd ACM SIGPLAN International Workshop on Machine Learning and Programming Languages},
  pages={10--19},
  year={2019}
}

@misc{zheng2025tritondistributedprogrammingoverlappingkernels,
      title={Triton-distributed: Programming Overlapping Kernels on Distributed AI Systems with the Triton Compiler}, 
      author={Size Zheng and Wenlei Bao and Qi Hou and Xuegui Zheng and Jin Fang and Chenhui Huang and Tianqi Li and Haojie Duanmu and Renze Chen and Ruifan Xu and Yifan Guo and Ningxin Zheng and Ziheng Jiang and Xinyi Di and Dongyang Wang and Jianxi Ye and Haibin Lin and Li-Wen Chang and Liqiang Lu and Yun Liang and Jidong Zhai and Xin Liu},
      year={2025},
      eprint={2504.19442},
      archivePrefix={arXiv},
      primaryClass={cs.DC},
      url={https://arxiv.org/abs/2504.19442}, 
}

@article{xu2021gspmd,
  title={GSPMD: general and scalable parallelization for ML computation graphs},
  author={Xu, Yuanzhong and Lee, HyoukJoong and Chen, Dehao and Hechtman, Blake and Huang, Yanping and Joshi, Rahul and Krikun, Maxim and Lepikhin, Dmitry and Ly, Andy and Maggioni, Marcello and others},
  journal={arXiv preprint arXiv:2105.04663},
  year={2021}
}

@article{paszke2019pytorch,
  title={Pytorch: An imperative style, high-performance deep learning library},
  author={Paszke, Adam and Gross, Sam and Massa, Francisco and Lerer, Adam and Bradbury, James and Chanan, Gregory and Killeen, Trevor and Lin, Zeming and Gimelshein, Natalia and Antiga, Luca and others},
  journal={Advances in neural information processing systems},
  volume={32},
  year={2019}
}

@inproceedings{aws-trainium,
author = {Bshara, Nafea},
title = {AWS Trainium: The Journey for Designing and Optimization Full Stack ML Hardware},
year = {2024},
isbn = {9798400703867},
publisher = {Association for Computing Machinery},
address = {New York, NY, USA},
url = {https://doi.org/10.1145/3620666.3655592},
doi = {10.1145/3620666.3655592},
booktitle = {Proceedings of the 29th ACM International Conference on Architectural Support for Programming Languages and Operating Systems, Volume 3},
pages = {4},
numpages = {1},
location = {La Jolla, CA, USA},
series = {ASPLOS '24}
}

@inproceedings{google-tpu,
  title={In-datacenter performance analysis of a tensor processing unit},
  author={Jouppi, Norman P and Young, Cliff and Patil, Nishant and Patterson, David and Agrawal, Gaurav and Bajwa, Raminder and Bates, Sarah and Bhatia, Suresh and Boden, Nan and Borchers, Al and others},
  booktitle={Proceedings of the 44th annual international symposium on computer architecture},
  pages={1--12},
  year={2017}
}

@article{nvidia2025blackwell,
  title={Blackwell GPU Architecture},
  author={Nvidia, RTX},
  journal={NVIDIA Corporation},
  year={2025}
}

@misc{jin2026eventtensorunifiedabstraction,
      title={Event Tensor: A Unified Abstraction for Compiling Dynamic Megakernel}, 
      author={Hongyi Jin and Bohan Hou and Guanjie Wang and Ruihang Lai and Jinqi Chen and Zihao Ye and Yaxing Cai and Yixin Dong and Xinhao Cheng and Zhihao Zhang and Yilong Zhao and Yingyi Huang and Lijie Yang and Jinchen Jiang and Gabriele Oliaro and Jianan Ji and Xupeng Miao and Vinod Grover and Todd C. Mowry and Zhihao Jia and Tianqi Chen},
      year={2026},
      eprint={2604.13327},
      archivePrefix={arXiv},
      primaryClass={cs.DC},
      url={https://arxiv.org/abs/2604.13327}, 
}

\end{document}